\pdfoutput=1
\documentclass[11pt,reqno]{amsart}
\usepackage{graphics,amsmath,amssymb,epsfig,stmaryrd,color,amsaddr}
\usepackage{subfigure,amsfonts,geometry,array}
\usepackage{graphicx}
\usepackage{float}
\usepackage{amsthm}
\usepackage{flafter}
\usepackage[svgnames]{xcolor}
\usepackage{bbm}
\usepackage{booktabs}
\usepackage{hyperref}

\usepackage{natbib}

\usepackage{enumerate}
\usepackage{multirow}

\makeatletter
\renewcommand\@endtheorem{\vvv@endmarker\endtrivlist\@endpefalse}
\newcommand\vvv@endmarker{%
  {\nobreak\hfil\penalty50
  \hskip2em\vadjust{}\nobreak\hfil\openbox
  \parfillskip=0pt \finalhyphendemerits=0 \par
  \penalty 10000 \parskip=0pt\noindent}\ignorespaces}
\makeatother

\theoremstyle{plain}

\newtheorem*{assumption*}{Assumption}
\newtheorem{assumption}{Assumption}

\newtheorem{lemma}{Lemma}

\newcommand\independent{\protect\mathpalette{\protect\independenT}{\perp}}
\def\independenT#1#2{\mathrel{\rlap{$#1#2$}\mkern2mu{#1#2}}}
\numberwithin{equation}{section}

\begin{document}
\nonstopmode

\title[Stochastic Frontier meets Breakdown Frontier]{Stochastic Frontier meets Breakdown Frontier}
\author{Santiago Acerenza}
\address{Universidad ORT Uruguay} 
\author{Francisco Rosas}
\address{Universidad ORT Uruguay and cinve} 
\noindent \date{\scriptsize{ The present version is as of \today. We thank Moriah B. Bostian, Subal C. Kumbhakar, Federico García-Suárez, and two anonymous readers for their valuable comments. 
}}

\begin{abstract}
This paper studies sensitivity analysis in stochastic frontier models by developing relaxations of the baseline assumptions imposed on the latent inefficiency and noise components, and characterize bounds for a benchmark technical-efficiency object under such relaxations. We then derive the associated breakdown frontier for conclusions about conditional technical efficiency and illustrate the procedure using a well-known dataset. We show the estimation and inference of the breakdown frontier. We also extend the analysis for widely used alternative specifications of the stochastic frontier error structure, that is, the Normal-Truncated Normal, Normal-Exponential, and Normal-Half Normal cases. Finally, we suggest avenues for extending this analysis under heteroskedasticity, the relaxation of input exogeneity, and applications using panel data. Code for empirical implementation is also provided.
\end{abstract}
 \maketitle
{\footnotesize \textbf{Keywords}: Stochastic Frontier Analysis, Technical inefficiency, Breakdown Frontier, Partial Identification.

 }

\newpage

\section{Introduction}

This paper studies sensitivity analysis in stochastic frontier models. Our guiding question is the following: how far can the baseline assumptions on the latent components of the model be relaxed while still sustaining a given conclusion about productive performance, conditional on the observed data? Since the seminal contributions of \cite{aigner1977formulation} and \cite{meeusen1977efficiency}, stochastic frontier analysis has been widely used to study productivity and efficiency across many economic sectors, including banking \citep{ferrier1990measuring,adams1999semiparametric,kumbhakar2005measuring,malikov2016cost}, healthcare \citep{zuckerman1994measuring,rosko2001cost,greene2004distinguishing,mutter2013investigating,comans2020cost}, and agriculture \citep{Bravo2007,Bravo2017,Trestini,Qushim2013,Ozden,Otieno,Gatti,Nwigwe,Qushim2018,Martinez,xin2025environmental,Lanfranco2013,Garcia2019,Garcia2022,Aguirre2024a,Aguirre2024b}. Stochastic frontier methods remain central in empirical work, and many important extensions have been developed since the original formulations in order to accommodate richer data structures and more flexible behavioural assumptions. For example, \cite{schmidt1984production,greene2005fixed,colombi2014closed} and \cite{kumbhakar2014technical} developed panel-data extensions, including decompositions into persistent and transitory inefficiency components. \cite{cornwell1990production,wang2002one,caudill1995frontier} considered models with determinants of inefficiency, while \cite{simar2017nonparametric,wang2024flexible,centorrino2024nonparametric,zheng2024robust} extended the methodology to semiparametric and nonparametric stochastic frontier models. See \cite{nguyen2022efficiency} for a thorough review of the literature and software implementation details.

Despite these important advances, sensitivity analysis per se remains largely absent from the stochastic frontier literature. Following \cite{masten2020inference}, we use the term sensitivity analysis to refer to a multidimensional framework that evaluates the robustness of empirical conclusions to simultaneous relaxations of identifying assumptions. Standard approaches begin by imposing assumptions under which a parameter becomes identified. By contrast, the breakdown-frontier perspective begins with a conclusion of interest and asks which combinations of assumptions are still sufficient to support that conclusion, given the observed data. This is formalized through a robust region and its associated breakdown frontier, which summarize the trade-offs between relaxing different assumptions while preserving the target conclusion.

This perspective is relevant for at least two reasons. First, all stochastic frontier models rely on assumptions about the latent inefficiency and noise components, even when those assumptions are relatively weak. Second, a breakdown-frontier analysis quantifies how much misspecification a benchmark model can tolerate before a given empirical conclusion is no longer supported by the data. In practice, this provides a useful complement to conventional stochastic frontier estimation: rather than replacing the null model, it shows how fragile or robust the corresponding conclusions are to controlled departures from its maintained assumptions.

A useful way to interpret the breakdown frontier is therefore as a \emph{robustness object}, rather than as a rule for selecting a single pair of sensitivity parameters. The frontier summarizes the full set of jointly admissible relaxations of the assumptions on the latent inefficiency and noise distributions that remain compatible with a benchmark conclusion about technical efficiency. In that sense, the contribution of the analysis is not to produce an alternative point estimate of efficiency, but to quantify the extent to which a benchmark stochastic frontier conclusion can withstand simultaneous misspecification in the latent components of the model.

Following \cite{masten2020inference}, the analysis requires six elements: (i) a parameter of interest, (ii) a set of baseline assumptions, (iii) a class of relaxations indexed by sensitivity parameters, (iv) bounds for the parameter of interest as a function of those sensitivity parameters, (v) a robust region and a breakdown frontier for a conclusion of interest, and (vi) estimation and inference procedures for these objects. We implement this framework with particular emphasis on the canonical model of \cite{aigner1977formulation}, although the same logic can be adapted, with appropriate modifications, to many of the stochastic frontier specifications used in practice.

A key conceptual distinction in our paper is the following. The technological frontier is estimated under a benchmark, or null, model. The sensitivity analysis is then conducted post-estimation on the implied technical-efficiency object. Thus, the paper does not re-estimate a different frontier under every possible relaxation. Instead, it studies how conclusions about technical efficiency vary as one weakens the maintained assumptions used to interpret the benchmark frontier model.

Compared with standard stochastic frontier analysis, the contribution of our approach is not to deliver a different or superior point estimator of technical efficiency under correct specification. Rather, the gain is that it adds a quantitative robustness layer to the conventional analysis. Standard SFA delivers a benchmark value of technical efficiency under maintained distributional assumptions. Our procedure asks how much simultaneous misspecification in the latent inefficiency and noise distributions can be tolerated before a given efficiency conclusion is no longer supported by the data. In this sense, the breakdown frontier complements standard SFA by separating the level of a benchmark efficiency estimate from the robustness of the conclusion drawn from it.

The rest of the paper proceeds as follows. Section 2 introduces the framework, defines the benchmark technical-efficiency object, and derives bounds under relaxed distributional assumptions. It then characterizes the robust region and the breakdown frontier in the normal--truncated normal case, and develops estimation and inference procedures for the resulting frontier. Section 3 provides an empirical illustration. The remaining sections discuss extensions and collect technical proofs and auxiliary results.

\section{Framework}

\subsection{Basic Set-up and the parameter of interest}

Let \(y_i\) be the logarithm of an output of interest. The canonical stochastic frontier model is
\[
y_i=f(\theta_x,x_i)+v_i-u_i,
\]
where \(f(\theta_x,x_i)\) is a known function of the inputs \(x_i\), up to parameters \(\theta_x\in\mathbb{R}^p\), \(v_i\) is a symmetric noise term, and \(u_i\) is the nonnegative latent inefficiency term, with \(u_i\geq 0\).

Define the composite error as
\[
\varepsilon_i=v_i-u_i.
\]
A standard object of interest in stochastic frontier analysis is
\[
E[\exp(-u_i)\mid \varepsilon_i].
\]
This object is naturally interpreted as \emph{technical efficiency} conditional on the  composite error. Accordingly, throughout the paper we focus on
\[
TE(\varepsilon):=E[\exp(-u_i)\mid \varepsilon_i=\varepsilon].
\]

\subsection{The structure of distributional assumptions and its relaxations}

It is common in the stochastic frontier literature to impose additional structure on \(v_i\) and \(u_i\) in order to identify \(TE(\varepsilon)\). A typical set of assumptions includes normality of \(v_i\) and some one-sided distribution for \(u_i\), potentially together with mean restrictions and heteroskedasticity.

The intuition behind the sensitivity analysis is the following. Suppose there exist two nonnegative constants \(b\) and \(c\) that measure, respectively, the distance between the assumed and true distributions of the noise term \(v_i\) and the inefficiency term \(u_i\). Then, in the context of the model above, suppose that under the null model the conditional technical efficiency for a given value of the composite error, \(\varepsilon\), is estimated to be \(\tau_0\in[0,1]\). The goal of the sensitivity analysis is to determine which combinations of \(c\) and \(b\) still sustain the conclusion that the relevant technical-efficiency object remains at least \(\tau_0\). This gives rise to the robust region (RR), namely the set of assumption relaxations under which the benchmark conclusion still holds. Conditional on a given value of one relaxation parameter, the corresponding threshold value of the other defines the breakdown frontier, illustrated in Figure \ref{fig1}.

\begin{figure}[H]
\centering
\includegraphics[scale=0.1]{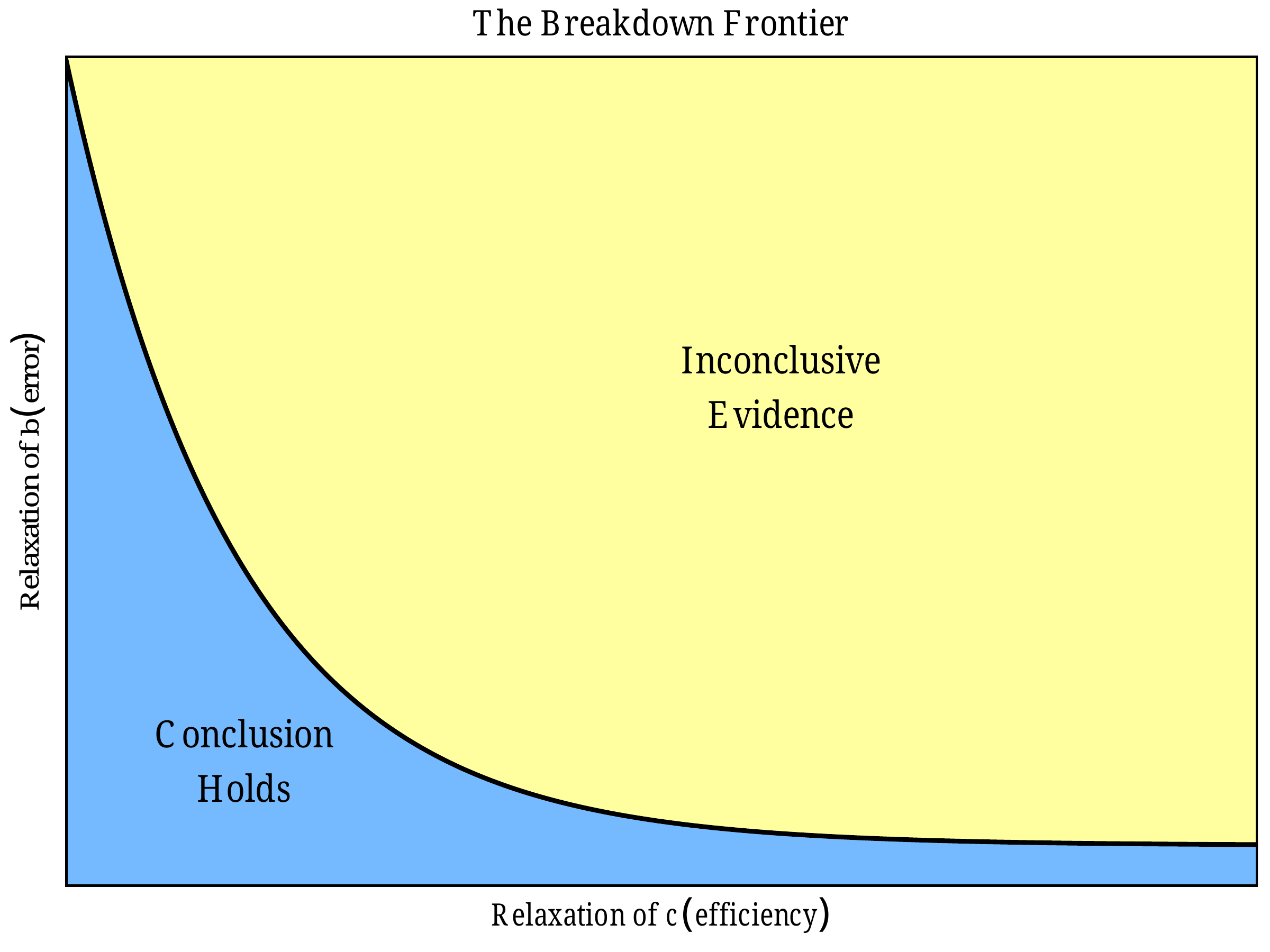}
\caption{The breakdown frontier. Source: Own elaboration based on \cite{masten2020inference}.}
\label{fig1}
\end{figure}

The blue region represents the set of combinations of \(c\) and \(b\), that is, violations of the baseline distributional assumptions, for which the benchmark conclusion about technical efficiency continues to hold. The yellow regions correspond to combinations for which the data no longer support that conclusion. The negative slope of the frontier reflects a trade-off between relaxing assumptions on the inefficiency component and relaxing assumptions on the noise component. Its curvature indicates how this trade-off varies across different degrees of misspecification.

Our starting point is therefore to index deviations from the null model by distances from the researcher-imposed benchmark distributions.

\begin{assumption}[Parametric distance]\label{AsPara}
Let \(f_{u_i}(u)\) denote the density of \(u_i\), and let \(f_{v_i\mid u_i}(v\mid u)\) denote the conditional density of \(v_i\) given \(u_i\). Further, let \(f_{u_i}(u;\theta_u)\) and \(f_{v_i\mid u_i}(v\mid u;\theta_v)\) denote researcher-specified parametric benchmark densities. Then
\begin{align*}
\sup_{u\geq 0}\big|f_{u_i}(u)-f_{u_i}(u;\theta_u)\big| &\leq c,\\
\sup_{v\in\mathbb{R}}\big|f_{v_i\mid u_i}(v\mid u)-f_{v_i\mid u_i}(v\mid u;\theta_v)\big| &\leq b
\qquad \forall u\geq 0.
\end{align*}
\end{assumption}

The sensitivity parameters \(b\) and \(c\) are therefore uniform distances between benchmark and true densities. As such, they are not probabilities, percentages, or quantities directly comparable to the technical-efficiency object itself. Their empirical magnitude depends on the scale and peak of the corresponding benchmark densities, which also helps explain why the numerical magnitudes of \(b\) and \(c\) need not be similar in applications. For this reason, the empirical analysis below complements the raw frontier representation with normalized interpretations of the sensitivity parameters when discussing magnitudes.

The case \(c=0\) and \(b=0\) corresponds to the null model. In particular, if \(v_i\) is normal and \(u_i\) is truncated normal, we recover the benchmark setting of \cite{aigner1977formulation}. Moving away from \((c,b)=(0,0)\) weakens those baseline distributional assumptions. For interpretative purposes, it can also be useful to work with normalized sensitivity parameters. A natural normalization is to scale each sensitivity parameter by the peak of its corresponding benchmark density, for example
\[
\tilde c=\frac{c}{\sup_{u\geq 0} f_u(u;\theta_u)},
\qquad
\tilde b=\frac{b}{\sup_{v\in\mathbb{R}} f_{v\mid u}(v\mid u;\theta_v)}.
\]
These normalized quantities do not alter the underlying robust region or breakdown frontier; they simply express the relaxations relative to the scale of the benchmark latent densities and can therefore facilitate interpretation.

Let \((x)_+:=\max\{x,0\}\). Under Assumption \ref{AsPara}, the true densities satisfy the pointwise bounds
\[
(f_{u_i}(u;\theta_u)-c)_+ \leq f_{u_i}(u) \leq f_{u_i}(u;\theta_u)+c,
\]
and
\[
(f_{v_i\mid u_i}(v\mid u;\theta_v)-b)_+ \leq f_{v_i\mid u_i}(v\mid u) \leq f_{v_i\mid u_i}(v\mid u;\theta_v)+b,
\qquad \forall (v,u)\in\mathbb{R}\times[0,\infty).
\]
It is important to stress that the clipped lower envelopes \((f_{u_i}(u;\theta_u)-c)_+\) and \((f_{v_i\mid u_i}(v\mid u;\theta_v)-b)_+\) are \emph{not} introduced as candidate densities. They are only lower envelope functions implied by the sup-norm relaxation together with nonnegativity of the true densities. Therefore, they do not need to integrate to one. Their role is purely to provide sharper lower bounds than the raw envelopes \(f_{u_i}(u;\theta_u)-c\) and \(f_{v_i\mid u_i}(v\mid u;\theta_v)-b\).

\subsection{Bounds for the parameter of interest under relaxed assumptions}

In this section, we derive bounds for the parameter of interest under relaxations of the benchmark assumptions. These bounds characterize the identified region for conditional technical efficiency under Assumption \ref{AsPara}.

For notational convenience, define the upper envelope functions
\[
\overline f_u(u;c):=f_{u_i}(u;\theta_u)+c,
\qquad
\overline f_{v\mid u}(v\mid u;b):=f_{v_i\mid u_i}(v\mid u;\theta_v)+b,
\]
and the clipped lower envelope functions
\[
\underline f_u(u;c):=(f_{u_i}(u;\theta_u)-c)_+,
\qquad
\underline f_{v\mid u}(v\mid u;b):=(f_{v_i\mid u_i}(v\mid u;\theta_v)-b)_+.
\]
Also define
\begin{align*}
\overline N(\varepsilon;c,b)
&:=
\int_0^\infty \exp(-u)\,
\overline f_{v\mid u}(\varepsilon+u\mid u;b)\,
\overline f_u(u;c)\,du,\\
\underline N(\varepsilon;c,b)
&:=
\int_0^\infty \exp(-u)\,
\underline f_{v\mid u}(\varepsilon+u\mid u;b)\,
\underline f_u(u;c)\,du,\\
\overline D(\varepsilon;c,b)
&:=
\int_0^\infty
\overline f_{v\mid u}(\varepsilon+u\mid u;b)\,
\overline f_u(u;c)\,du,\\
\underline D(\varepsilon;c,b)
&:=
\int_0^\infty
\underline f_{v\mid u}(\varepsilon+u\mid u;b)\,
\underline f_u(u;c)\,du.
\end{align*}

\begin{lemma}\label{LemmaIdentifiedset}
Suppose Assumption \ref{AsPara} holds. If \(\underline D(\varepsilon;c,b)>0\), then
\begin{equation}\label{EqBP5}
\frac{\underline N(\varepsilon;c,b)}{\overline D(\varepsilon;c,b)}
\leq
E[\exp(-u_i)\mid \varepsilon_i=\varepsilon]
\leq
\frac{\overline N(\varepsilon;c,b)}{\underline D(\varepsilon;c,b)}.
\end{equation}
\end{lemma}

Thus, for any given realization \(\varepsilon\) of the composite error, Lemma \ref{LemmaIdentifiedset} delivers lower and upper bounds for the corresponding conditional technical-efficiency object. These bounds are the basis for the robust-region and breakdown-frontier analysis developed below. The proof of Lemma \ref{LemmaIdentifiedset} can be found in Appendix \ref{Appproofs}

\subsection{An application to the Normal--Truncated Normal model: using the bounds to define the breakdown frontier for a conclusion of interest}\label{sec2.4}

Suppose that \(u_i\) and \(v_i\) are independent, and assume that
\[
u_i\sim TN(\mu,\sigma_u^2),
\qquad
v_i\sim N(0,\sigma_v^2).
\]
Then
\begin{align*}
f_{u_i}(u;\theta_u)
&=
\frac{1}{\sqrt{2\pi\sigma_u^2}}
\frac{\exp\!\left(-\frac{(u-\mu)^2}{2\sigma_u^2}\right)}
{1-\Phi\!\left(\frac{-\mu}{\sigma_u}\right)},
\\
f_{v_i\mid u_i}(v\mid u;\theta_v)
&=
\frac{1}{\sqrt{2\pi\sigma_v^2}}
\exp\!\left(-\frac{v^2}{2\sigma_v^2}\right).
\end{align*}

Define
\[
\int_{0}^{\infty}\exp(-u)\,f_{v_i\mid u_i}(\varepsilon+u\mid u;\theta_v)\,du
\equiv \Delta_1(\varepsilon;\sigma_v),
\]
\[
\int_{0}^{\infty}\exp(-u)\,f_{v_i\mid u_i}(\varepsilon+u\mid u;\theta_v)f_{u_i}(u;\theta_u)\,du
\equiv \Delta_1(\varepsilon;\sigma_v,\sigma_u,\mu),
\]
\[
\int_{0}^{\infty}\exp(-u)\,f_{u_i}(u;\theta_u)\,du
\equiv \Delta_1(\sigma_u,\mu),
\]
\[
\int_{0}^{\infty}f_{v_i\mid u_i}(\varepsilon+u\mid u;\theta_v)\,du
\equiv \Delta_2(\varepsilon;\sigma_v),
\]
\[
\int_{0}^{\infty}f_{v_i\mid u_i}(\varepsilon+u\mid u;\theta_v)f_{u_i}(u;\theta_u)\,du
\equiv \Delta_2(\varepsilon;\sigma_v,\sigma_u,\mu).
\]
Closed-form expressions for these objects are given in Lemma \ref{LemmaAuxiliar}. We extend these calculations to alternative null models for \(u_i\) in Appendix \ref{AppAltNulls}.

For the upper-envelope objects, independence and the unit-mass property of benchmark densities imply
\begin{align*}
\overline N(\varepsilon;c,b)
&=
\Delta_1(\varepsilon;\sigma_v,\sigma_u,\mu)
+c\,\Delta_1(\varepsilon;\sigma_v)
+b\,\Delta_1(\sigma_u,\mu)
+bc,\\
\overline D(\varepsilon;c,b)
&=
\Delta_2(\varepsilon;\sigma_v,\sigma_u,\mu)
+c\,\Delta_2(\varepsilon;\sigma_v)
+b.
\end{align*}
For the clipped lower-envelope objects, define
\begin{align*}
\underline N_{TN}(\varepsilon;c,b)
&:=
\int_0^\infty \exp(-u)
\big(f_{v_i\mid u_i}(\varepsilon+u\mid u;\theta_v)-b\big)_+
\big(f_{u_i}(u;\theta_u)-c\big)_+\,du,\\
\underline D_{TN}(\varepsilon;c,b)
&:=
\int_0^\infty
\big(f_{v_i\mid u_i}(\varepsilon+u\mid u;\theta_v)-b\big)_+
\big(f_{u_i}(u;\theta_u)-c\big)_+\,du.
\end{align*}

\subsubsection{The identified region}

By direct application of equation \eqref{EqBP5},
\begin{equation}\label{NormalIdentifiedset}
\frac{
\underline N_{TN}(\varepsilon;c,b)
}{
\Delta_2(\varepsilon;\sigma_v,\sigma_u,\mu)
+c\,\Delta_2(\varepsilon;\sigma_v)
+b
}
\leq
E[\exp(-u_i)\mid \varepsilon_i=\varepsilon]
\leq
\frac{
\Delta_1(\varepsilon;\sigma_v,\sigma_u,\mu)
+c\,\Delta_1(\varepsilon;\sigma_v)
+b\,\Delta_1(\sigma_u,\mu)
+bc
}{
\underline D_{TN}(\varepsilon;c,b)
},
\end{equation}
provided \(\underline D_{TN}(\varepsilon;c,b)>0\).

\subsubsection{The breakdown frontier}

We now define the breakdown frontier in this context. We also introduce the robust region, that is, the set of relaxations under which a benchmark conclusion about technical efficiency remains valid.

Suppose that under the benchmark stochastic frontier model, corresponding to \(c=0\) and \(b=0\), one obtains
\[
E[\exp(-u_i)\mid \varepsilon_i=\varepsilon]=\tau_0,
\]
for a fixed \(\tau_0\in[0,1]\). The sensitivity analysis then starts from the conclusion
\[
E[\exp(-u_i)\mid \varepsilon_i=\varepsilon]\geq \tau_0.
\]
Relative to the baseline assumptions, what are the weakest assumptions that still allow us to draw this conclusion, given the observed distribution of the data? Since larger values of \(c\) and \(b\) correspond to weaker assumptions, the question becomes: what are the largest values of \(c\) and \(b\) such that we can still conclude that
\[
E[\exp(-u_i)\mid \varepsilon_i=\varepsilon]\geq \tau_0?
\]

Following \cite{masten2020inference}, the robust region is the set of relaxations for which the lower bound in \eqref{NormalIdentifiedset} remains above \(\tau_0\). Hence,
\begin{equation}\label{RobustRegion}
RR=
\left\{
(c,b)\in[0,\infty)^2:
\frac{
\underline N_{TN}(\varepsilon;c,b)
}{
\Delta_2(\varepsilon;\sigma_v,\sigma_u,\mu)
+c\,\Delta_2(\varepsilon;\sigma_v)
+b
}
\geq \tau_0
\right\}.
\end{equation}

Note that when \(c=0\) and \(b=0\),
\[
\underline N_{TN}(\varepsilon;0,0)=\Delta_1(\varepsilon;\sigma_v,\sigma_u,\mu),
\qquad
\overline D(\varepsilon;0,0)=\Delta_2(\varepsilon;\sigma_v,\sigma_u,\mu),
\]
so
\[
\tau_0=
\frac{\Delta_1(\varepsilon;\sigma_v,\sigma_u,\mu)}
{\Delta_2(\varepsilon;\sigma_v,\sigma_u,\mu)}.
\]
Hence, the benchmark model is contained in the robust region and \(RR\neq\emptyset\).

The robust region and its boundary should be interpreted as robustness summaries, not as devices for selecting a single preferred pair \((c,b)\). Their role is to characterize which combinations of misspecification in the inefficiency and noise distributions remain compatible with the technical-efficiency conclusion of interest. In particular, the frontier is useful for comparing robustness across different values of the composite error, different benchmark conclusions \(\tau_0\), and different benchmark specifications of the latent components.

The breakdown frontier is the boundary of the robust region, namely
\begin{equation}\label{Breakdownfrontier}
BF(\tau_0)=
\left\{
(c,b)\in[0,\infty)^2:
\underline N_{TN}(\varepsilon;c,b)
-
\tau_0\Big(
\Delta_2(\varepsilon;\sigma_v,\sigma_u,\mu)
+c\,\Delta_2(\varepsilon;\sigma_v)
+b
\Big)
=0
\right\}.
\end{equation}

Unlike the representation obtained from the raw linear envelopes, the clipped lower-envelope formulation generally does not admit a closed-form expression for \(b\) as an explicit function of \(c\). Accordingly, the breakdown frontier is characterized implicitly by \eqref{Breakdownfrontier} and can be computed numerically for each fixed value of \(c\).

\subsection{The Normal--Truncated Normal model continued: estimation and inference for the breakdown frontier}

In this section, we study estimation and inference for the breakdown frontier defined in Section \ref{sec2.4}. We extend the analysis to alternative specifications of \(u_i\) in Appendix \ref{AppAltNulls_25_detail}. A key conceptual point is that the technological frontier is estimated under the benchmark, or ``null,'' model, that is, the model with \(c=0\) and \(b=0\). The sensitivity analysis is then conducted post-estimation on the corresponding technical-efficiency object.

Accordingly, we estimate
\[
\vartheta=(\theta_x,\mu,\sigma_v,\sigma_u)\in\mathbb{R}^{p+3}.
\]
For an i.i.d.\ sample of size \(n\), the log-likelihood is
\begin{align}
LnL
&=
\sum_{i=1}^n \ln f(y_i,x_i;\vartheta)
\nonumber\\
&=
\sum_{i=1}^n
\Bigg(
-\frac{1}{2}\ln(2\pi)
-\ln\Big(\sqrt{\sigma_v^2+\sigma_u^2}\Big)
-\ln\Phi\Big(\frac{\mu}{\sigma_u}\Big)
\nonumber\\
&\qquad\qquad
+\ln\Phi\Bigg(
\frac{
\left(1-\frac{\sigma_u^2}{\sigma_u^2+\sigma_v^2}\right)\mu
-
\frac{\sigma_u^2}{\sigma_u^2+\sigma_v^2}\big(y_i-f(\theta_x,x_i)\big)
}{
\sqrt{(\sigma_u^2+\sigma_v^2)\frac{\sigma_u^2}{\sigma_u^2+\sigma_v^2}\left(1-\frac{\sigma_u^2}{\sigma_u^2+\sigma_v^2}\right)}
}
\Bigg)
-\frac{1}{2}\Bigg(
\frac{y_i-f(\theta_x,x_i)+\mu}{\sqrt{\sigma_v^2+\sigma_u^2}}
\Bigg)^2
\Bigg).
\label{LogLikelihood}
\end{align}
Let
\[
\widehat{\vartheta}_n=(\widehat{\theta}_{xn},\widehat{\mu}_n,\widehat{\sigma}_{vn},\widehat{\sigma}_{un})
\]
denote the corresponding maximum-likelihood estimator.

Under standard regularity conditions,
\[
\sqrt{n}\big(\widehat{\vartheta}_n-\vartheta\big)\xrightarrow{d}N(\boldsymbol{0},\mathcal{V}_{\vartheta}),
\]
where the asymptotic variance has the usual sandwich form
\[
\mathcal{V}_{\vartheta}
=
\mathcal{H}_{\vartheta}^{-1}\mathcal{J}_{\vartheta}\mathcal{H}_{\vartheta}^{-1},
\]
with
\[
\mathcal{H}_{\vartheta}
=
-E\left[
\frac{\partial^2}{\partial \vartheta\,\partial \vartheta'}\ln f(y_i,x_i;\vartheta)
\right],
\qquad
\mathcal{J}_{\vartheta}
=
E\left[
\left(
\frac{\partial}{\partial \vartheta}\ln f(y_i,x_i;\vartheta)
\right)
\left(
\frac{\partial}{\partial \vartheta}\ln f(y_i,x_i;\vartheta)
\right)'
\right].
\]
Under correct specification, the information equality implies \(\mathcal{J}_{\vartheta}=\mathcal{H}_{\vartheta}\), so that \(\mathcal{V}_{\vartheta}=\mathcal{H}_{\vartheta}^{-1}\).

A consistent estimator of the asymptotic variance is
\[
\widehat{\mathcal{V}}_{\vartheta,n}
=
\widehat{\mathcal{H}}_{\vartheta,n}^{-1}
\widehat{\mathcal{J}}_{\vartheta,n}
\widehat{\mathcal{H}}_{\vartheta,n}^{-1},
\]
where
\[
\widehat{\mathcal{H}}_{\vartheta,n}
=
-\frac{1}{n}\sum_{i=1}^n
\left[
\frac{\partial^2}{\partial \widehat{\vartheta}_n\,\partial \widehat{\vartheta}_n'}\ln f(y_i,x_i;\widehat{\vartheta}_n)
\right],
\qquad
\widehat{\mathcal{J}}_{\vartheta,n}
=
\frac{1}{n}\sum_{i=1}^n
\widehat{s}_i\widehat{s}_i',
\]
with
\[
\widehat{s}_i=
\frac{\partial}{\partial \vartheta}\ln f(y_i,x_i;\widehat{\vartheta}_n).
\]

To conduct inference on the breakdown frontier under the clipped-envelope formulation, we smooth two nonsmooth objects: the positive-part operator used in the clipped density envelopes, and the nonnegativity restriction on the frontier parameter \(b\).

For \(x\geq 0\) and \(a\geq 0\), define
\begin{equation}\label{eq:smoothclip}
S_{\kappa}(x;a)
:=
\kappa\log\!\left(1+\exp\!\left(\frac{x-a}{\kappa}\right)\right)
-
\kappa\log\!\left(1+\exp\!\left(\frac{-a}{\kappa}\right)\right),
\qquad \kappa>0.
\end{equation}
This operator satisfies \(S_{\kappa}(0;a)=0\), \(S_{\kappa}(x;a)\geq 0\), and
\[
S_{\kappa}(x;a)\to (x-a)_+
\qquad\text{as}\qquad \kappa\downarrow 0.
\]

Next, introduce an unrestricted scalar \(\beta\in\mathbb R\), and map it into a nonnegative relaxation parameter through the smooth positive transformation
\begin{equation}\label{eq:softplusb}
B_{\eta}(\beta)
:=
\eta\log\!\left(1+\exp\!\left(\frac{\beta}{\eta}\right)\right),
\qquad \eta>0.
\end{equation}
For each fixed \(\eta>0\), \(B_{\eta}(\beta)\) is strictly positive and continuously differentiable in \(\beta\), with derivative
\begin{equation}\label{eq:softplus_derivative}
B_{\eta}'(\beta)
=
\frac{1}{1+\exp(-\beta/\eta)}.
\end{equation}
Moreover,
\[
B_{\eta}(\beta)\to \beta_+
\qquad\text{as}\qquad \eta\downarrow 0.
\]

Using \eqref{eq:smoothclip}, define the smoothed lower-envelope functions
\[
\underline f_{u,\kappa}(u;c):=S_{\kappa}(f_{u_i}(u;\theta_u);c),
\qquad
\underline f_{v\mid u,\kappa}(v\mid u;b):=S_{\kappa}(f_{v_i\mid u_i}(v\mid u;\theta_v);b),
\]
and the associated smoothed lower-bound numerator
\begin{equation}\label{eq:smoothedNTN}
\underline N_{TN,\kappa}(\varepsilon;c,b;\vartheta)
:=
\int_0^\infty
\exp(-u)\,
\underline f_{v\mid u,\kappa}(\varepsilon+u\mid u;b)\,
\underline f_{u,\kappa}(u;c)\,du.
\end{equation}

Recall that the upper-denominator term in the Normal--Truncated Normal case is
\[
\overline D_{TN}(\varepsilon;c,b;\vartheta)
=
\Delta_2(\varepsilon;\sigma_v,\sigma_u,\mu)
+
c\,\Delta_2(\varepsilon;\sigma_v)
+
b.
\]
Hence, for fixed \((c,\tau_0)\), define the smoothed frontier equation
\begin{equation}\label{eq:Gkappa}
G_{\kappa}(b;c,\tau_0,\vartheta)
:=
\underline N_{TN,\kappa}(\varepsilon;c,b;\vartheta)
-
\tau_0\,\overline D_{TN}(\varepsilon;c,b;\vartheta).
\end{equation}

To impose \(b\geq 0\) smoothly, compose \(G_{\kappa}\) with the map \(B_{\eta}\). Define
\begin{equation}\label{eq:Gkappaeta}
\widetilde G_{\kappa,\eta}(\beta;c,\tau_0,\vartheta)
:=
G_{\kappa}(B_{\eta}(\beta);c,\tau_0,\vartheta).
\end{equation}
For each \(c\), let \(\beta_{\kappa,\eta}(c,\tau_0)\) denote the unique solution in \(\mathbb R\) to
\begin{equation}\label{eq:implicitbeta}
\widetilde G_{\kappa,\eta}(\beta;c,\tau_0,\vartheta)=0,
\end{equation}
and define the corresponding smooth breakdown frontier by
\begin{equation}\label{eq:bkappaeta}
b_{\kappa,\eta}(c,\tau_0):=B_{\eta}\big(\beta_{\kappa,\eta}(c,\tau_0)\big).
\end{equation}
Its plug-in estimator is defined by
\begin{equation}\label{eq:implicitbetahat}
\widetilde G_{\kappa_n,\eta_n}\big(\widehat\beta_{n,\kappa_n,\eta_n}(c,\tau_0);c,\tau_0,\widehat{\vartheta}_n\big)=0,
\end{equation}
with
\begin{equation}\label{eq:bhatkappaeta}
\widehat b_{n,\kappa_n,\eta_n}(c,\tau_0)
=
B_{\eta_n}\big(\widehat\beta_{n,\kappa_n,\eta_n}(c,\tau_0)\big),
\end{equation}
where \(\kappa_n>0\) and \(\eta_n>0\) are deterministic smoothing sequences.

For fixed \((\kappa,\eta)\), \(\widetilde G_{\kappa,\eta}(\beta;c,\tau_0,\vartheta)\) is continuously differentiable in both \(\beta\) and \(\vartheta\). Therefore, provided
\[
\frac{\partial}{\partial \beta}
\widetilde G_{\kappa,\eta}\big(\beta_{\kappa,\eta}(c,\tau_0);c,\tau_0,\vartheta\big)\neq 0,
\]
the implicit-function theorem and the delta method imply
\begin{equation}\label{eq:asymptoticbeta}
\sqrt{n}\Big(\widehat\beta_{n,\kappa,\eta}(c,\tau_0)-\beta_{\kappa,\eta}(c,\tau_0)\Big)
\xrightarrow{d}
N\big(0,\sigma^2_{\beta_{\kappa,\eta}(c,\tau_0)}\big),
\end{equation}
where
\begin{equation}\label{eq:Dbeta}
\boldsymbol{D}_{\beta_{\kappa,\eta}(c,\tau_0)}
=
-
\left(
\frac{\partial}{\partial \beta}
\widetilde G_{\kappa,\eta}\big(\beta_{\kappa,\eta}(c,\tau_0);c,\tau_0,\vartheta\big)
\right)^{-1}
\frac{\partial}{\partial \vartheta'}
\widetilde G_{\kappa,\eta}\big(\beta_{\kappa,\eta}(c,\tau_0);c,\tau_0,\vartheta\big),
\end{equation}
and
\begin{equation}\label{eq:varbeta}
\sigma^2_{\beta_{\kappa,\eta}(c,\tau_0)}
=
\boldsymbol{D}_{\beta_{\kappa,\eta}(c,\tau_0)}
\mathcal{V}_{\vartheta}
\boldsymbol{D}_{\beta_{\kappa,\eta}(c,\tau_0)}'.
\end{equation}

Since \(b_{\kappa,\eta}(c,\tau_0)=B_{\eta}(\beta_{\kappa,\eta}(c,\tau_0))\), another application of the delta method yields
\begin{equation}\label{eq:asymptoticb}
\sqrt{n}\Big(\widehat b_{n,\kappa,\eta}(c,\tau_0)-b_{\kappa,\eta}(c,\tau_0)\Big)
\xrightarrow{d}
N\big(0,\sigma^2_{b_{\kappa,\eta}(c,\tau_0)}\big),
\end{equation}
with influence derivative
\begin{equation}\label{eq:Dbeta_to_Db}
\boldsymbol{D}_{b_{\kappa,\eta}(c,\tau_0)}
=
B_{\eta}'\big(\beta_{\kappa,\eta}(c,\tau_0)\big)\,
\boldsymbol{D}_{\beta_{\kappa,\eta}(c,\tau_0)},
\end{equation}
and asymptotic variance
\begin{equation}\label{eq:varbkappaeta}
\sigma^2_{b_{\kappa,\eta}(c,\tau_0)}
=
\boldsymbol{D}_{b_{\kappa,\eta}(c,\tau_0)}
\mathcal{V}_{\vartheta}
\boldsymbol{D}_{b_{\kappa,\eta}(c,\tau_0)}'.
\end{equation}

A feasible estimator of \(\sigma^2_{b_{\kappa,\eta}(c,\tau_0)}\) is obtained by replacing population quantities with their plug-in analogues:
\[
\widehat{\sigma}^2_{b_{\kappa,\eta}(c,\tau_0),n}
=
\widehat{\boldsymbol{D}}_{b_{\kappa,\eta}(c,\tau_0),n}
\widehat{\mathcal{V}}_{\vartheta,n}
\widehat{\boldsymbol{D}}_{b_{\kappa,\eta}(c,\tau_0),n}'.
\]

Because the frontier is now characterized through the smooth implicit equation \eqref{eq:implicitbeta}, and because under the maintained condition that, for each fixed \((c,\tau_0,\vartheta)\), the map
\[
\beta \mapsto \widetilde G_{\kappa,\eta}(\beta;c,\tau_0,\vartheta)
\]
is strictly monotone, the corresponding solution \(\beta_{\kappa,\eta}(c,\tau_0)\) is uniquely defined. Standard nonparametric bootstrap methods can therefore be applied directly to \(\widehat\beta_{n,\kappa,\eta}(c,\tau_0)\) for fixed \((\kappa,\eta)\), and the resulting bootstrap draws can then be mapped into the frontier scale through \eqref{eq:bhatkappaeta}. In particular, the bootstrap avoids the need to deal directly with the nonsmooth positivity restriction \(b\geq 0\).

\subsubsection{What happens when \((\kappa,\eta)\) shrink?}

The doubly smoothed frontier \(b_{\kappa,\eta}(c,\tau_0)\) approximates the original clipped-envelope frontier as \((\kappa,\eta)\downarrow (0,0)\). Let \(b(c,\tau_0)\) denote the solution to the unsmoothed frontier equation
\[
\underline N_{TN}(\varepsilon;c,b;\vartheta)
-
\tau_0\,\overline D_{TN}(\varepsilon;c,b;\vartheta)
=0.
\]
The next lemma states that, if both smoothing parameters vanish slowly enough and the target frontier is strictly positive, the plug-in estimator based on the doubly smoothed problem remains consistent for the original clipped-envelope frontier.

\begin{lemma}\label{LemmaSmooth}
Assume that, for each fixed \(c\), the clipped frontier equation admits a unique solution \(b(c,\tau_0)\in(0,\infty)\), and that \(G_{\kappa}(b;c,\tau_0,\vartheta)\) is strictly decreasing in \(b\) on a neighborhood of \(b(c,\tau_0)\), uniformly for sufficiently small \(\kappa\). If \(\kappa_n\downarrow 0\), \(\eta_n\downarrow 0\), \(\sqrt{n}\kappa_n\to\infty\), and \(\sqrt{n}\eta_n\to\infty\), then
\[
\widehat b_{n,\kappa_n,\eta_n}(c,\tau_0)\xrightarrow{p} b(c,\tau_0).
\]
\end{lemma}
The proof of Lemma \ref{LemmaSmooth} can be found in Appendix \ref{Appproofs}

\section{Empirical Illustration}

In this section, we illustrate the procedure developed in the previous sections using the data analysed in \cite{nguyen2022efficiency}, which contain output and input information for rice producers in the Philippines.\footnote{The dataset is available in the \texttt{frontier} package. See \url{https://frontier.r-forge.r-project.org/} and \url{https://www.rdocumentation.org/packages/frontier/versions/1.1-8/topics/riceProdPhil}.} The dataset contains information on 43 rice producers in Tarlac, Philippines, from 1990 to 1997. From the original dataset, we use one output \(y_i\) (freshly threshed rice in tons) and three inputs \(x_i\): planted area (in hectares), labour used (man-days of family and hired labour), and fertilizer used (active ingredients in kilograms).

We estimate the benchmark stochastic frontier under three alternative one-sided distributions for inefficiency: Normal--Truncated Normal (TN), Normal--Exponential (EXP), and Normal--Half-Normal (HN). For each benchmark specification, we compute the implied benchmark technical-efficiency object
\[
TE_0(\varepsilon)=E[\exp(-u_i)\mid \varepsilon_i=\varepsilon]
\]
at several values of the composite error \(\varepsilon\), and then evaluate the corresponding smooth breakdown frontier  over a grid of values of \(c\). We report results for two types of conclusions: (i) the benchmark conclusion \(TE(\varepsilon)\geq TE_0(\varepsilon)\), and (ii) fixed absolute thresholds such as \(TE(\varepsilon)\geq 0.70\) and \(TE(\varepsilon)\geq 0.80\).

In what follows, we report the breakdown frontier in its original \((c,b)\) scale, since this is the object directly implied by the identification analysis. However, the numerical magnitudes of \(b\) and \(c\) should be interpreted with caution, because they are measured in units of density rather than in probabilities or percentages. For this reason, one may alternatively read the frontier through normalized parameters, obtained by scaling \(b\) and \(c\) by the peaks of the corresponding benchmark latent densities. This normalization does not replace the original frontier, but provides a complementary way of assessing how large a given relaxation is relative to the benchmark model.

The empirical results deliver three main findings. First, when the conclusion of interest is set equal to the benchmark technical-efficiency value, that is, \(TE(\varepsilon)\geq TE_0(\varepsilon)\), the estimated smooth frontier collapses to zero throughout the considered grid of values of \(c\). This occurs for all values of \(\varepsilon\) that we examine and for all three benchmark specifications. Hence, in this application, the benchmark conclusion is locally fragile to simultaneous relaxations of the assumptions on the inefficiency and noise distributions.

Second, although the benchmark frontier is degenerate, the benchmark technical-efficiency level itself varies substantially with \(\varepsilon\). Table~\ref{Table1} reports the benchmark values \(\tau_0=TE_0(\varepsilon)\). Under the TN and EXP benchmarks, these values are nearly identical throughout the support of \(\varepsilon\); under the HN benchmark they are somewhat smaller, especially in the lower tail of the composite-error distribution. For example, at the median residual quantile \(p50\), the benchmark values are approximately \(0.835\) under TN and EXP, and \(0.803\) under HN. At the lower quantile \(p10\), they are approximately \(0.570\) under TN and EXP, and \(0.538\) under HN.

Third, once we replace the benchmark conclusion by more moderate absolute thresholds, nontrivial frontiers emerge. In particular, for conclusions such as \(TE(\varepsilon)\geq 0.70\), the frontier is positive for a non-negligible range of values of \(c\), especially around \(\varepsilon=0\), \(p50\), \(p75\), and \(p90\). By contrast, for \(TE(\varepsilon)\geq 0.80\), the admissible region shrinks sharply, and for \(TE(\varepsilon)\geq 0.90\) it becomes nearly degenerate. This pattern is consistent with the interpretation of the breakdown frontier as a robustness summary: weaker conclusions admit larger robust regions, while stronger conclusions survive only under much smaller deviations from the benchmark model.

These results also clarify what is gained relative to standard stochastic frontier analysis. A conventional SFA exercise would stop at reporting the benchmark values \(TE_0(\varepsilon)\) under each maintained distributional specification. The breakdown-frontier analysis adds a second layer of information: it reveals whether those benchmark conclusions are robust or fragile to simultaneous misspecification in the latent inefficiency and noise distributions. In our application, this distinction is empirically important. Although the benchmark values \(TE_0(\varepsilon)\) vary substantially across \(\varepsilon\) and benchmark models, the associated benchmark conclusion \(TE(\varepsilon)\geq TE_0(\varepsilon)\) is essentially degenerate in robustness terms. By contrast, weaker conclusions such as \(TE(\varepsilon)\geq 0.70\) and, to a lesser extent, \(TE(\varepsilon)\geq 0.80\), do admit nontrivial robust regions. Thus, the gain from the breakdown-frontier approach is not a new benchmark estimate of efficiency, but a quantitative characterization of how robust that benchmark conclusion is.

Figure~\ref{fig3} compares the point-estimated smooth breakdown frontier across benchmark specifications at \(\varepsilon=p50\) under the benchmark conclusion \(TE(\varepsilon)\geq TE_0(\varepsilon)\). The three curves coincide at zero throughout the grid, indicating that, in this case, the benchmark conclusion is equally fragile under TN, EXP, and HN. This figure is useful because it shows that the fragility of the benchmark conclusion is not driven by one particular parametric specification.

Figure~\ref{fig4} keeps the benchmark specification fixed at TN and the composite-error value fixed at \(\varepsilon=p50\), and compares the frontier across alternative technical-efficiency conclusions. The benchmark and \(0.90\) curves are essentially degenerate, while the frontier is clearly positive for \(TE(\varepsilon)\geq 0.70\) and, to a lesser extent, for \(TE(\varepsilon)\geq 0.80\). This shows that the empirical content of the sensitivity analysis becomes more informative when one evaluates conclusions that are economically relevant but less demanding than the benchmark equality-based threshold.

To illustrate inference for the smooth breakdown frontier, Figure~\ref{fig5} reports pointwise cluster-bootstrap confidence intervals for a representative nondegenerate case: the TN benchmark, evaluated at \(\varepsilon=p50\), for the conclusion \(TE(\varepsilon)\geq 0.70\). We focus on this case because the benchmark frontier is degenerate in our application, whereas this weaker conclusion yields a nontrivial frontier. The bootstrap intervals are informative over the region where the frontier is positive and shrink toward zero as the admissible robust region vanishes. This confirms that bootstrap-based inference for the smooth frontier is feasible in empirically relevant nondegenerate settings.

Tables~\ref{Table2A} and \ref{Table2B} report point-estimated frontier values at \(c=0\) and \(c=0.10\), respectively. They reinforce the same message. When the inefficiency distribution is not relaxed (\(c=0\)), the benchmark conclusion remains degenerate, while the conclusions \(TE(\varepsilon)\geq 0.70\) and \(TE(\varepsilon)\geq 0.80\) still admit positive relaxations in the noise distribution. Once the inefficiency distribution is relaxed to \(c=0.10\), the admissible relaxation in the noise density shrinks further, especially for the stricter \(0.80\) threshold. This is exactly the trade-off summarized by the breakdown frontier.

Taken together, the empirical exercise delivers two substantive messages. On the one hand, benchmark technical-efficiency conclusions obtained from standard parametric stochastic frontier models can be highly fragile once one allows for simultaneous misspecification in the latent distributions. On the other hand, the breakdown-frontier framework remains informative because it reveals which weaker, economically interpretable conclusions about technical efficiency still admit nontrivial robust regions, and how those regions differ across benchmark specifications and values of the composite error. A frontier close to zero should therefore not be read as saying that the values of \(b\) or \(c\) are “small” in an absolute sense, but rather as indicating that the conclusion of interest is fragile: once the inefficiency distribution is even slightly relaxed, little or no additional relaxation in the noise distribution can be sustained while preserving the same conclusion.

\begin{figure}[H]
\centering
\includegraphics[width=0.70\textwidth,height=0.38\textheight,keepaspectratio]{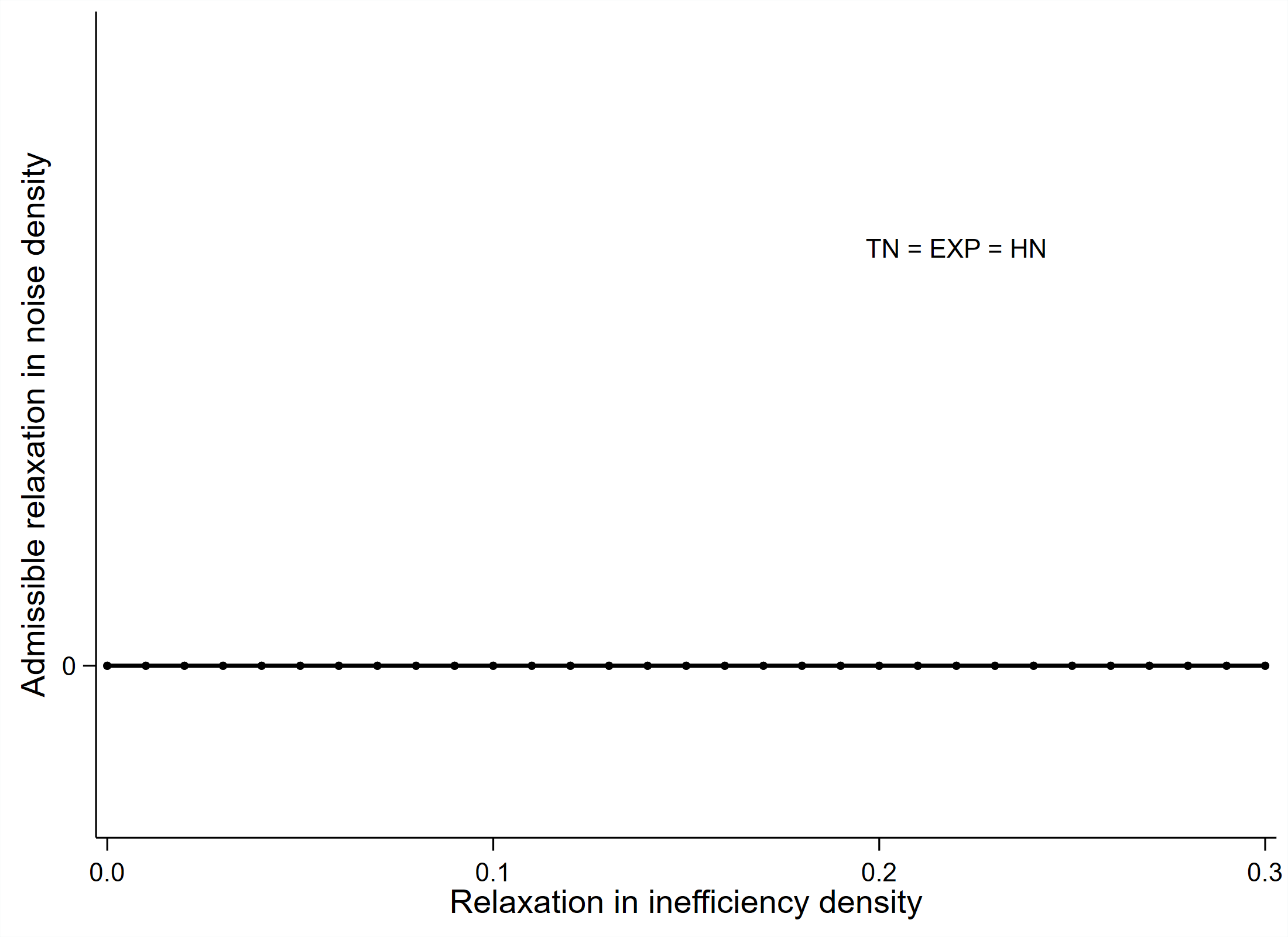}
\caption{Point-estimated smooth breakdown frontier \(b_{\kappa}(c,\tau_0)\) as a function of the relaxation parameter \(c\), evaluated at \(\varepsilon=p50\) under the benchmark conclusion \(TE(\varepsilon)\geq TE_0(\varepsilon)\). The vertical axis reports the maximum admissible relaxation in the noise density, \(b_{\kappa}(c,\tau_0)\), that remains compatible with the conclusion for each fixed value of \(c\). TN, EXP, and HN denote the Normal--Truncated Normal, Normal--Exponential, and Normal--Half-Normal benchmark models, respectively. In this application, all three frontiers coincide at zero, indicating that the benchmark conclusion is locally fragile across benchmark specifications.}
\label{fig3}
\end{figure}

\begin{figure}[H]
\centering
\includegraphics[width=0.76\textwidth,height=0.42\textheight,keepaspectratio]{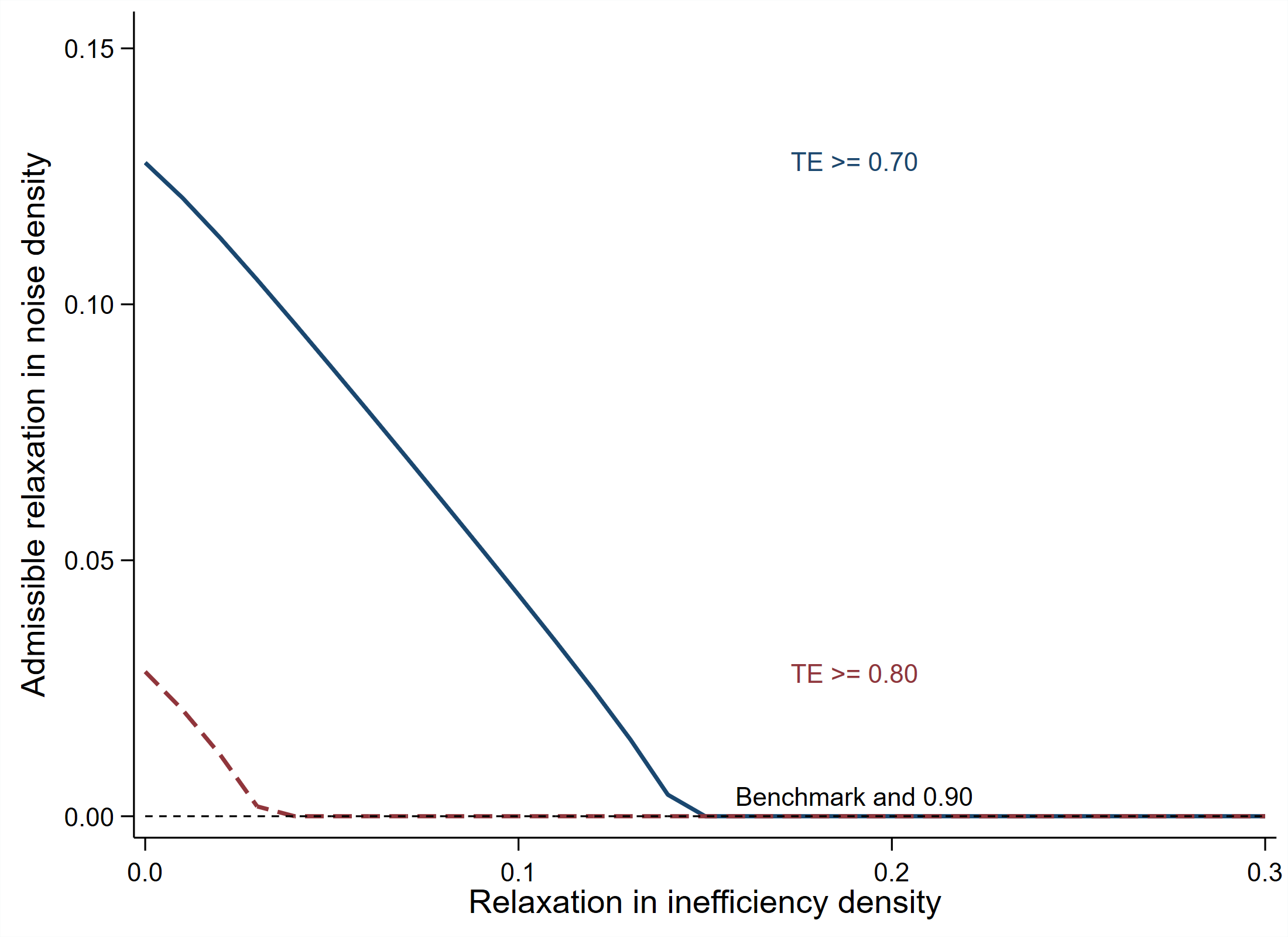}
\caption{Point-estimated smooth breakdown frontier \(b_{\kappa}(c,\tau_0)\) under the TN benchmark, evaluated at \(\varepsilon=p50\), for alternative technical-efficiency conclusions. The benchmark curve corresponds to \(TE(\varepsilon)\geq TE_0(\varepsilon)\), while the other curves correspond to absolute thresholds \(TE(\varepsilon)\geq 0.70\), \(0.80\), and \(0.90\). The figure shows that nontrivial robust regions emerge for weaker conclusions, whereas stricter and benchmark conclusions are nearly degenerate in this application.}
\label{fig4}
\end{figure}

\begin{figure}[H]
\centering
\includegraphics[width=0.76\textwidth,height=0.42\textheight,keepaspectratio]{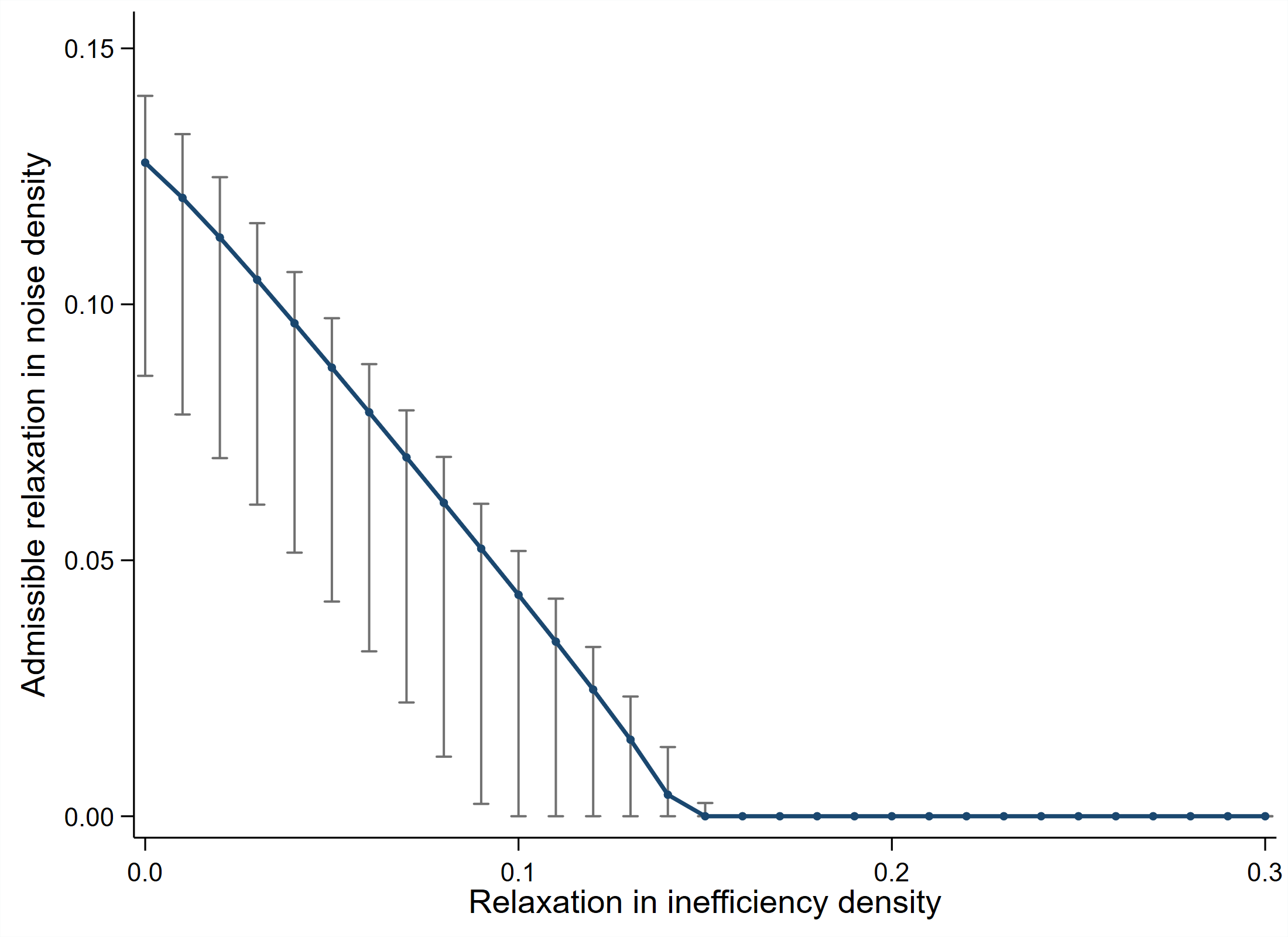}
\caption{Point estimate and pointwise cluster-bootstrap confidence interval for the smooth breakdown frontier \(b_{\kappa}(c,\tau_0)\) under the TN benchmark, evaluated at \(\varepsilon=p50\), for the conclusion \(TE(\varepsilon)\geq 0.70\). The solid line is the point estimate, and the shaded region reports the pointwise bootstrap interval based on 200 cluster-bootstrap replications at the producer level. The interval is informative over the range of values of \(c\) for which the frontier is positive, and collapses toward zero as the admissible robust region vanishes.}
\label{fig5}
\end{figure}
\begin{center}
\refstepcounter{table}\label{Table1}
\textbf{Table \thetable.} Benchmark technical-efficiency values \(\tau_0=TE_0(\varepsilon)\) by benchmark model and value of the composite error. TN and EXP produce nearly identical benchmark values throughout, while HN yields somewhat lower values, especially in the lower tail of \(\varepsilon\).\\[0.5em]
{\small
\begin{tabular}{lccc}
\toprule
 & TN & EXP & HN \\\\
\midrule
0 &  0.898 &  0.898 &  0.887 \\\\
p10 &  0.570 &  0.570 &  0.538 \\\\
p25 &  0.732 &  0.733 &  0.685 \\\\
p50 &  0.835 &  0.835 &  0.803 \\\\
p75 &  0.889 &  0.889 &  0.875 \\\\
p90 &  0.912 &  0.912 &  0.906 \\\\
\bottomrule
\end{tabular}
}

\end{center}

\begin{center}
\refstepcounter{table}\label{Table2A}
\textbf{Table \thetable.} Point-estimated values of the smooth breakdown frontier \(b_{\kappa}(c,\tau_0)\) at \(c=0\), by benchmark model, value of \(\varepsilon\), and conclusion of interest. The table reports the maximum admissible relaxation in the noise density when the inefficiency distribution is not relaxed. The benchmark conclusion yields degenerate frontiers, while the conclusions \(TE(\varepsilon)\geq 0.70\) and \(TE(\varepsilon)\geq 0.80\) generate positive but substantially smaller admissible regions.\\[0.5em]
{\small
\begin{tabular}{lccccccccc}
\toprule
 & \multicolumn{3}{c}{Benchmark} & \multicolumn{3}{c}{\(TE(\varepsilon)\geq 0.70\)} & \multicolumn{3}{c}{\(TE(\varepsilon)\geq 0.80\)} \\\\
\cmidrule(lr){2-4}\cmidrule(lr){5-7}\cmidrule(lr){8-10}
\(\varepsilon\) & TN & EXP & HN & TN & EXP & HN & TN & EXP & HN \\\\
\midrule
0 &  0.000 &  0.000 &  0.000 &  0.158 &  0.159 &  0.119 &  0.071 &  0.071 &  0.049 \\\\
p10 &  0.000 &  0.000 &  0.000 &  0.000 &  0.000 &  0.000 &  0.000 &  0.000 &  0.000 \\\\
p25 &  0.000 &  0.000 &  0.000 &  0.015 &  0.016 &  0.000 &  0.000 &  0.000 &  0.000 \\\\
p50 &  0.000 &  0.000 &  0.000 &  0.128 &  0.128 &  0.096 &  0.028 &  0.028 &  0.000 \\\\
p75 &  0.000 &  0.000 &  0.000 &  0.167 &  0.167 &  0.129 &  0.072 &  0.072 &  0.049 \\\\
p90 &  0.000 &  0.000 &  0.000 &  0.127 &  0.127 &  0.088 &  0.061 &  0.061 &  0.040 \\\\
\bottomrule
\end{tabular}
}

\end{center}

\begin{center}
\refstepcounter{table}\label{Table2B}
\textbf{Table \thetable.} Point-estimated values of the smooth breakdown frontier \(b_{\kappa}(c,\tau_0)\) at \(c=0.10\), by benchmark model, value of \(\varepsilon\), and conclusion of interest. The table reports the maximum admissible relaxation in the noise density when the inefficiency distribution is relaxed by \(c=0.10\). Relative to Table~\ref{Table2A}, the admissible values of \(b\) shrink further, especially for stricter technical-efficiency thresholds.\\[0.5em]
{\small
\begin{tabular}{lccccccccc}
\toprule
 & \multicolumn{3}{c}{Benchmark} & \multicolumn{3}{c}{\(TE(\varepsilon)\geq 0.70\)} & \multicolumn{3}{c}{\(TE(\varepsilon)\geq 0.80\)} \\\\
\cmidrule(lr){2-4}\cmidrule(lr){5-7}\cmidrule(lr){8-10}
\(\varepsilon\) & TN & EXP & HN & TN & EXP & HN & TN & EXP & HN \\\\
\midrule
0 &  0.000 &  0.000 &  0.000 &  0.107 &  0.107 &  0.059 &  0.018 &  0.018 &  0.000 \\\\
p10 &  0.000 &  0.000 &  0.000 &  0.000 &  0.000 &  0.000 &  0.000 &  0.000 &  0.000 \\\\
p25 &  0.000 &  0.000 &  0.000 &  0.000 &  0.000 &  0.000 &  0.000 &  0.000 &  0.000 \\\\
p50 &  0.000 &  0.000 &  0.000 &  0.043 &  0.043 &  0.000 &  0.000 &  0.000 &  0.000 \\\\
p75 &  0.000 &  0.000 &  0.000 &  0.108 &  0.108 &  0.060 &  0.010 &  0.010 &  0.000 \\\\
p90 &  0.000 &  0.000 &  0.000 &  0.091 &  0.091 &  0.048 &  0.024 &  0.024 &  0.000 \\\\
\bottomrule
\end{tabular}
}

\end{center}

\section{Further remarks}

\subsection{Allowing for heteroskedasticity}

Heteroskedasticity can be incorporated naturally into the framework, since the benchmark, or null, model can be estimated under alternative assumptions on the variances and means of the latent components. In particular, one may allow the distribution of \(u_i\) and/or \(v_i\) to depend on observable covariates through variance or location functions, and then carry out the same post-estimation sensitivity analysis on the resulting technical-efficiency object.

Under the revised clipped-envelope formulation, the main conceptual structure remains unchanged: the benchmark model is first estimated under the maintained parametric specification, and the sensitivity analysis is then performed post-estimation. The main difference is that the lower-envelope objects generally become
\[
\big(f_{u_i}(u;\theta_u)-c\big)_+,
\qquad
\big(f_{v_i\mid u_i}(v\mid u;\theta_v)-b\big)_+,
\]
or their smooth approximations, so that the relevant lower-bound integrals typically no longer admit simple closed forms. As a result, once heteroskedasticity is introduced, the frontier becomes a function of a richer parameter vector and is generally characterized through implicit equations and numerical integration rather than closed-form algebra.

\subsection{Relaxing exogeneity of the inputs}

The previous analysis focuses on distributional assumptions for the latent inefficiency and noise components. However, the breakdown-frontier framework can also be extended to incorporate relaxations of input exogeneity. One possible route would be to introduce sensitivity parameters that bound deviations from exogeneity through quantities such as
\[
|E(u_i\mid x_i)-E(u_i)|,
\qquad
|V(u_i\mid x_i)-V(u_i)|.
\]
Such bounds would capture departures from the benchmark assumption that the inefficiency component is conditionally unrelated to the observed inputs.

Extending the current framework in this direction would require re-deriving the corresponding bounds for the technical-efficiency object under these additional relaxations. In particular, the resulting envelopes would no longer be indexed solely by distributional distances for \(u_i\) and \(v_i\), but also by parameters governing deviations from exogeneity. Consequently, both the identified region and the associated breakdown frontier would have to be redefined in a higher-dimensional sensitivity space.

\subsection{Extension to panel-data models}

The empirical illustration focuses on the cross-sectional setting. Nevertheless, the procedure can be adapted to panel-data models in at least two directions. First, the cross-sectional analysis can be implemented year by year, thereby generating a sequence of breakdown frontiers over time. Second, in the presence of fixed effects, where the model becomes
\[
y_{i,t}=f(\theta_x,x_{i,t})+v_{i,t}-u_{i,t}+c_i,
\]
the framework can be implemented using first differences,
\[
\Delta y_{i,t,t+1}
=
\Delta f(\theta_x,x_{i,t,t+1})
+
\Delta v_{i,t,t+1}
-
\Delta u_{i,t,t+1},
\]
in the spirit of \cite{chen2014consistent,belotti2012consistent}.

Under the revised sensitivity-analysis logic, the benchmark panel-data frontier would again be estimated under the null model, while the post-estimation sensitivity analysis would be conducted through clipped or smooth-clipped envelope bounds for the relevant transformed latent components. Future research could further extend this methodology to stochastic-frontier panel-data models with persistent and transitory inefficiency components, as in \cite{colombi2014closed} and \cite{kumbhakar2014technical}.

\section{Conclusions}

This paper develops a sensitivity-analysis framework for stochastic frontier models. Starting from a benchmark stochastic frontier specification, we derive bounds for a conditional technical-efficiency object under controlled relaxations of the baseline distributional assumptions on the inefficiency and noise components. On the basis of these bounds, we characterize the corresponding robust region and breakdown frontier for conclusions about technical efficiency. We provide a complete analysis of estimation and inference of the breakdown frontier. We also show an empirical illustration using a well-known dataset and make the corresponding code available for practitioners.

A central message of the paper is that robustness in stochastic frontier analysis can be summarized through a transparent trade-off between relaxations of the assumptions imposed on the inefficiency component and relaxations of the assumptions imposed on the noise component. For a given benchmark conclusion about technical efficiency, the breakdown frontier identifies the combinations of deviations from the null model that are still compatible with that conclusion. In turn, the robust region provides a direct measure of how much misspecification the benchmark model can tolerate before the conclusion is no longer sustained by the data.

Because the sensitivity parameters are defined as distances between densities, their interpretation is inherently relative to the benchmark latent distributions. For this reason, normalized versions of \(b\) and \(c\) can be useful as a complement to the raw frontier when discussing empirical magnitudes.

The empirical illustration shows that the benchmark conclusion \(TE(\varepsilon)\geq TE_0(\varepsilon)\) can be highly fragile, with the corresponding frontier collapsing to zero in our application. At the same time, weaker and economically interpretable conclusions, such as \(TE(\varepsilon)\geq 0.70\) or \(TE(\varepsilon)\geq 0.80\), generate nontrivial downward-sloping frontiers. This illustrates how the breakdown-frontier approach distinguishes between the level of benchmark technical efficiency and the robustness of the conclusion drawn from it.

A key conceptual feature of the analysis is that the technological frontier is estimated under the benchmark, or null, model, while the sensitivity analysis is conducted post-estimation on the implied technical-efficiency object. Thus, the contribution of the paper is not to propose a different frontier estimator under every relaxation, but rather to formalize how conclusions drawn from the benchmark frontier vary once the maintained assumptions are progressively weakened.

More broadly, the results suggest that the gain from the breakdown-frontier approach, relative to standard stochastic frontier analysis, lies in its ability to quantify robustness rather than to replace the benchmark frontier estimator. Standard SFA provides a benchmark estimate of technical efficiency under maintained assumptions. The breakdown frontier complements that information by showing which efficiency conclusions remain valid once those assumptions are weakened, and by how much. This distinction between benchmark efficiency levels and robustness of efficiency conclusions is, in our view, the main added value of the approach.

We show how our analysis can be extended to typical specifications of the null model error structure other than the Normal-Truncated Normal distributions. These are the Normal-Exponential and the Normal-Half Normal cases, for which we not only derive the breakdown frontier and present the inference procedures, but also conduct the empirical illustration under these alternative distributional specifications, which in turn allows the comparison of the breakdown frontier behavior under different assumed distributions.

Finally, we propose how our analysis can be extended to contemplate some features relevant for practitioners. Firstly, the ability that the assumed variances allow for heteroskedasticity. Secondly, to incorporate relaxations of the input exogeneity by means of introducing additional sensitivity parameters that bound the deviations from the exogeneity assumptions, which will increase the dimensionality of the frontier space. Thirdly, we propose ideas for adapting the procedures to panel data applications, and consider some avenues of future research on other specifications of the stochastic frontier analysis such as the persistent and transitory inefficiency components. 

\renewcommand{\refname}{References}
\bibliographystyle{chicago}
\bibliography{references_all}
\appendix

\section{Proof of the Main results}\label{Appproofs}

\begin{proof}[Proof of Lemma \ref{LemmaIdentifiedset}]
Let
\[
\overline f_u(u;c):=f_{u_i}(u;\theta_u)+c,
\qquad
\underline f_u(u;c):=(f_{u_i}(u;\theta_u)-c)_+,
\]
and
\[
\overline f_{v\mid u}(v\mid u;b):=f_{v_i\mid u_i}(v\mid u;\theta_v)+b,
\qquad
\underline f_{v\mid u}(v\mid u;b):=(f_{v_i\mid u_i}(v\mid u;\theta_v)-b)_+.
\]
By Assumption \ref{AsPara}, together with nonnegativity of the true densities, we have
\begin{equation}\label{EqAs1}
\underline f_u(u;c)\leq f_{u_i}(u)\leq \overline f_u(u;c),
\qquad
\underline f_{v\mid u}(v\mid u;b)\leq f_{v_i\mid u_i}(v\mid u)\leq \overline f_{v\mid u}(v\mid u;b)
\quad \forall (v,u)\in\mathbb{R}\times[0,\infty).
\end{equation}

Note also that
\begin{align}
f_{\varepsilon_i}(\varepsilon)
&=
\frac{\partial P(\varepsilon_i\leq \varepsilon)}{\partial \varepsilon}
=
\int_{0}^{\infty} f_{v_i\mid u_i}(\varepsilon+u\mid u)f_{u_i}(u)\,du,
\nonumber\\
f_{\varepsilon_i\mid u_i}(\varepsilon\mid u)
&=
\frac{\partial P(\varepsilon_i\leq \varepsilon\mid u_i=u)}{\partial \varepsilon}
=
f_{v_i\mid u_i}(\varepsilon+u\mid u).
\label{EqBP1}
\end{align}

Applying \eqref{EqAs1} inside \eqref{EqBP1} yields
\begin{align}
\underline D(\varepsilon;c,b)
:=
\int_0^\infty
\underline f_{v\mid u}(\varepsilon+u\mid u;b)\,
\underline f_u(u;c)\,du
&\leq
f_{\varepsilon_i}(\varepsilon)
\leq
\int_0^\infty
\overline f_{v\mid u}(\varepsilon+u\mid u;b)\,
\overline f_u(u;c)\,du
=:
\overline D(\varepsilon;c,b),
\label{EqBP2}
\end{align}
and, for every \(u\geq 0\),
\begin{equation}\label{EqBP2b}
\underline f_{v\mid u}(\varepsilon+u\mid u;b)
\leq
f_{\varepsilon_i\mid u_i}(\varepsilon\mid u)
\leq
\overline f_{v\mid u}(\varepsilon+u\mid u;b).
\end{equation}

By Bayes' rule,
\begin{equation}\label{EqBP3}
f_{u_i\mid \varepsilon_i}(u\mid \varepsilon)
=
\frac{f_{\varepsilon_i\mid u_i}(\varepsilon\mid u)f_{u_i}(u)}{f_{\varepsilon_i}(\varepsilon)}.
\end{equation}

Combining \eqref{EqAs1}, \eqref{EqBP2}, \eqref{EqBP2b}, and \eqref{EqBP3}, and using \(\underline D(\varepsilon;c,b)>0\), we obtain
\begin{equation}\label{EqBP4upper}
f_{u_i\mid \varepsilon_i}(u\mid \varepsilon)
\leq
\frac{
\overline f_{v\mid u}(\varepsilon+u\mid u;b)\,
\overline f_u(u;c)
}{
\underline D(\varepsilon;c,b)
},
\end{equation}
and
\begin{equation}\label{EqBP4lower}
f_{u_i\mid \varepsilon_i}(u\mid \varepsilon)
\geq
\frac{
\underline f_{v\mid u}(\varepsilon+u\mid u;b)\,
\underline f_u(u;c)
}{
\overline D(\varepsilon;c,b)
}.
\end{equation}

Therefore,
\begin{align}
E[\exp(-u_i)\mid \varepsilon_i=\varepsilon]
&=
\int_0^\infty \exp(-u)\,f_{u_i\mid \varepsilon_i}(u\mid \varepsilon)\,du
\nonumber\\
&\leq
\frac{
\int_0^\infty \exp(-u)\,
\overline f_{v\mid u}(\varepsilon+u\mid u;b)\,
\overline f_u(u;c)\,du
}{
\underline D(\varepsilon;c,b)
}
=
\frac{\overline N(\varepsilon;c,b)}{\underline D(\varepsilon;c,b)},
\nonumber\\
E[\exp(-u_i)\mid \varepsilon_i=\varepsilon]
&\geq
\frac{
\int_0^\infty \exp(-u)\,
\underline f_{v\mid u}(\varepsilon+u\mid u;b)\,
\underline f_u(u;c)\,du
}{
\overline D(\varepsilon;c,b)
}
=
\frac{\underline N(\varepsilon;c,b)}{\overline D(\varepsilon;c,b)}.
\label{EqBP5proof}
\end{align}
This proves the result.

\begin{proof}[Proof of Lemma \ref{LemmaSmooth}]
 Let
\[
G_{0}(b;c,\tau_0,\vartheta)
:=
\underline N_{TN}(\varepsilon;c,b;\vartheta)
-
\tau_0\,\overline D_{TN}(\varepsilon;c,b;\vartheta),
\]
and let \(b(c,\tau_0)\in(0,\infty)\) denote its unique zero. For each \((\kappa,\eta)\), define
\[
\widetilde G_{\kappa,\eta}(\beta;c,\tau_0,\vartheta)
:=
G_{\kappa}(B_\eta(\beta);c,\tau_0,\vartheta),
\]
and let \(\beta_{\kappa,\eta}(c,\tau_0)\) be its unique zero, with
\[
b_{\kappa,\eta}(c,\tau_0)
=
B_\eta(\beta_{\kappa,\eta}(c,\tau_0)).
\]
The proof proceeds in three steps.

\medskip
\noindent
\textbf{Step 1: \(G_{\kappa}(b;c,\tau_0,\vartheta)\to G_0(b;c,\tau_0,\vartheta)\) pointwise in \(b\).}

By construction,
\[
S_\kappa(x;a)\to (x-a)_+
\qquad\text{as}\qquad \kappa\downarrow 0,
\]
for every \(x\ge 0\), \(a\ge 0\). Moreover,
\[
0\le S_\kappa(x;a)\le x
\qquad\text{for all }x\ge 0,\ a\ge 0.
\]
Hence, for each fixed \((c,b)\), the integrand in
\[
\underline N_{TN,\kappa}(\varepsilon;c,b;\vartheta)
=
\int_0^\infty
\exp(-u)\,
S_\kappa(f_{v_i\mid u_i}(\varepsilon+u\mid u;\theta_v);b)\,
S_\kappa(f_{u_i}(u;\theta_u);c)\,du
\]
converges pointwise to the integrand in \(\underline N_{TN}(\varepsilon;c,b;\vartheta)\), and is dominated by
\[
\exp(-u)\,f_{v_i\mid u_i}(\varepsilon+u\mid u;\theta_v)\,f_{u_i}(u;\theta_u),
\]
which is integrable under the null model. Therefore, by dominated convergence,
\[
\underline N_{TN,\kappa}(\varepsilon;c,b;\vartheta)
\to
\underline N_{TN}(\varepsilon;c,b;\vartheta)
\qquad\text{as}\qquad \kappa\downarrow 0.
\]
Since \(\overline D_{TN}(\varepsilon;c,b;\vartheta)\) does not depend on \(\kappa\), it follows that
\[
G_{\kappa}(b;c,\tau_0,\vartheta)\to G_0(b;c,\tau_0,\vartheta)
\]
pointwise in \(b\).

\medskip
\noindent
\textbf{Step 2: \(b_{\kappa,\eta}(c,\tau_0)\to b(c,\tau_0)\).}

Because
\[
B_\eta(\beta)\to \beta_+
\qquad\text{as}\qquad \eta\downarrow 0,
\]
and \(b(c,\tau_0)>0\) by assumption, there exists
\[
\beta(c,\tau_0)=b(c,\tau_0)
\]
such that \(B_\eta(\beta(c,\tau_0))\to b(c,\tau_0)\). Combining this with Step 1 yields
\[
\widetilde G_{\kappa,\eta}(\beta;c,\tau_0,\vartheta)
=
G_{\kappa}(B_\eta(\beta);c,\tau_0,\vartheta)
\to
G_0(\beta_+;c,\tau_0,\vartheta)
\]
for each fixed \(\beta\).

Since \(b(c,\tau_0)\in(0,\infty)\) is the unique zero of \(G_0(b;c,\tau_0,\vartheta)\), and \(G_\kappa(b;c,\tau_0,\vartheta)\) is strictly decreasing in \(b\) on a neighborhood of \(b(c,\tau_0)\), uniformly for sufficiently small \(\kappa\), standard continuity-of-roots arguments imply
\[
b_{\kappa,\eta}(c,\tau_0)\to b(c,\tau_0)
\qquad\text{as}\qquad (\kappa,\eta)\downarrow(0,0).
\]

\medskip
\noindent
\textbf{Step 3: \(\widehat b_{n,\kappa_n,\eta_n}(c,\tau_0)-b_{\kappa_n,\eta_n}(c,\tau_0)\xrightarrow{p}0\).}

For each \(n\),
\[
\widetilde G_{\kappa_n,\eta_n}(\widehat\beta_{n,\kappa_n,\eta_n}(c,\tau_0);c,\tau_0,\widehat\vartheta_n)=0,
\]
while
\[
\widetilde G_{\kappa_n,\eta_n}(\beta_{\kappa_n,\eta_n}(c,\tau_0);c,\tau_0,\vartheta)=0.
\]
Because \((\kappa_n,\eta_n)\) are fixed at each \(n\), the map
\[
(\beta,\vartheta)\mapsto \widetilde G_{\kappa_n,\eta_n}(\beta;c,\tau_0,\vartheta)
\]
is continuously differentiable. Since \(\widehat\vartheta_n\xrightarrow{p}\vartheta\), a first-order expansion around \((\beta_{\kappa_n,\eta_n}(c,\tau_0),\vartheta)\) yields
\begin{align*}
0
&=
\widetilde G_{\kappa_n,\eta_n}(\widehat\beta_{n,\kappa_n,\eta_n};c,\tau_0,\widehat\vartheta_n)
-
\widetilde G_{\kappa_n,\eta_n}(\beta_{\kappa_n,\eta_n};c,\tau_0,\vartheta)
\\
&=
\frac{\partial}{\partial \beta}
\widetilde G_{\kappa_n,\eta_n}(\widetilde\beta_n;c,\tau_0,\widetilde\vartheta_n)
\big(
\widehat\beta_{n,\kappa_n,\eta_n}-\beta_{\kappa_n,\eta_n}
\big)
+
\frac{\partial}{\partial \vartheta'}
\widetilde G_{\kappa_n,\eta_n}(\widetilde\beta_n;c,\tau_0,\widetilde\vartheta_n)
(\widehat\vartheta_n-\vartheta),
\end{align*}
for some intermediate values \(\widetilde\beta_n\) and \(\widetilde\vartheta_n\).

By the maintained derivative condition,
\[
\frac{\partial}{\partial \beta}
\widetilde G_{\kappa_n,\eta_n}(\widetilde\beta_n;c,\tau_0,\widetilde\vartheta_n)
\]
is bounded away from zero in probability for sufficiently large \(n\). Hence
\[
\widehat\beta_{n,\kappa_n,\eta_n}(c,\tau_0)-\beta_{\kappa_n,\eta_n}(c,\tau_0)=o_p(1).
\]
Since \(B_{\eta_n}\) is Lipschitz with derivative bounded by one,
\[
\big|
\widehat b_{n,\kappa_n,\eta_n}(c,\tau_0)-b_{\kappa_n,\eta_n}(c,\tau_0)
\big|
\le
\big|
\widehat\beta_{n,\kappa_n,\eta_n}(c,\tau_0)-\beta_{\kappa_n,\eta_n}(c,\tau_0)
\big|,
\]
and therefore
\[
\widehat b_{n,\kappa_n,\eta_n}(c,\tau_0)-b_{\kappa_n,\eta_n}(c,\tau_0)=o_p(1).
\]

Finally,
\[
\widehat b_{n,\kappa_n,\eta_n}(c,\tau_0)-b(c,\tau_0)
=
\big(
\widehat b_{n,\kappa_n,\eta_n}(c,\tau_0)-b_{\kappa_n,\eta_n}(c,\tau_0)
\big)
+
\big(
b_{\kappa_n,\eta_n}(c,\tau_0)-b(c,\tau_0)
\big).
\]
The first term is \(o_p(1)\) by Step 3, and the second converges to zero by Step 2 because \(\kappa_n\downarrow 0\) and \(\eta_n\downarrow 0\). Therefore,
\[
\widehat b_{n,\kappa_n,\eta_n}(c,\tau_0)\xrightarrow{p}b(c,\tau_0).
\]
This proves the result.
\end{proof}

\end{proof}
\section{Auxiliary results}\label{Appextra}
\footnotesize
\subsection{Normal--truncated normal integrals}

\begin{lemma}\label{LemmaAuxiliar}
Let
\begin{align*}
f_{u_i}(u;\theta_u)
&=
\frac{1}{\sqrt{2\pi\sigma_u^2}}
\frac{\exp\!\left(-\frac{(u-\mu)^2}{2\sigma_u^2}\right)}
{1-\Phi\!\left(-\frac{\mu}{\sigma_u}\right)},
\\
f_{v_i\mid u_i}(v\mid u;\theta_v)
&=
\frac{1}{\sqrt{2\pi\sigma_v^2}}
\exp\!\left(-\frac{v^2}{2\sigma_v^2}\right).
\end{align*}
Then
\begin{align*}
\Delta_1(\varepsilon;\sigma_v)
&\equiv
\int_{0}^{\infty}\exp(-u)\,f_{v_i\mid u_i}(\varepsilon+u\mid u;\theta_v)\,du
\\
&=
\exp\!\left(
-\frac{\varepsilon^2}{2\sigma_v^2}
+
\frac{(-\varepsilon-\sigma_v^2)^2}{2\sigma_v^2}
\right)
\left[
1-\Phi\!\left(
\sqrt{2}\,
\frac{-(-\varepsilon-\sigma_v^2)}{\sqrt{2\sigma_v^2}}
\right)
\right],
\\[0.4em]
\Delta_1(\varepsilon;\sigma_v,\sigma_u,\mu)
&\equiv
\int_{0}^{\infty}\exp(-u)\,f_{v_i\mid u_i}(\varepsilon+u\mid u;\theta_v)f_{u_i}(u;\theta_u)\,du
\\
&=
\sqrt{\frac{2\sigma_v^2\sigma_u^2}{\sigma_u^2+\sigma_v^2}}
\,
\frac{
\exp\!\left(
-\frac{\varepsilon^2}{2\sigma_v^2}
-\frac{\mu^2}{2\sigma_u^2}
+
\frac{\sigma_u^2+\sigma_v^2}{2\sigma_v^2\sigma_u^2}
\left[
\frac{\mu \sigma_v^2-\varepsilon \sigma_u^2-\sigma_v^2\sigma_u^2}{\sigma_u^2+\sigma_v^2}
\right]^2
\right)
}{
\sqrt{2}\,\sigma_v\sigma_u\sqrt{2\pi}\left(1-\Phi\!\left(-\frac{\mu}{\sigma_u}\right)\right)
}
\\
&\qquad\times
\left[
1-\Phi\!\left(
\sqrt{2}
\left(
\frac{
-\frac{\mu \sigma_v^2-\varepsilon \sigma_u^2-\sigma_v^2\sigma_u^2}{\sigma_u^2+\sigma_v^2}
}{
\sqrt{\frac{2\sigma_v^2\sigma_u^2}{\sigma_u^2+\sigma_v^2}}
}
\right)
\right)
\right],
\\[0.4em]
\Delta_1(\sigma_u,\mu)
&\equiv
\int_{0}^{\infty}\exp(-u)\,f_{u_i}(u;\theta_u)\,du
\\
&=
\frac{
\exp\!\left(
-\frac{\mu^2}{2\sigma_u^2}
+
\frac{(\mu-\sigma_u^2)^2}{2\sigma_u^2}
\right)
}{
1-\Phi\!\left(-\frac{\mu}{\sigma_u}\right)
}
\left[
1-\Phi\!\left(
\sqrt{2}
\left(
\frac{-(\mu-\sigma_u^2)}{\sqrt{2\sigma_u^2}}
\right)
\right)
\right],
\\[0.4em]
\Delta_2(\varepsilon;\sigma_v)
&\equiv
\int_{0}^{\infty}f_{v_i\mid u_i}(\varepsilon+u\mid u;\theta_v)\,du
=
\left[
1-\Phi\!\left(
\sqrt{2}
\left(
\frac{-(-\varepsilon)}{\sqrt{2\sigma_v^2}}
\right)
\right)
\right],
\\[0.4em]
\Delta_2(\varepsilon;\sigma_v,\sigma_u,\mu)
&\equiv
\int_{0}^{\infty}f_{v_i\mid u_i}(\varepsilon+u\mid u;\theta_v)f_{u_i}(u;\theta_u)\,du
\\
&=
\sqrt{\frac{2\sigma_v^2\sigma_u^2}{\sigma_u^2+\sigma_v^2}}
\,
\frac{
\exp\!\left(
-\frac{\varepsilon^2}{2\sigma_v^2}
-\frac{\mu^2}{2\sigma_u^2}
+
\frac{\sigma_u^2+\sigma_v^2}{2\sigma_v^2\sigma_u^2}
\left[
\frac{\mu \sigma_v^2-\varepsilon \sigma_u^2}{\sigma_u^2+\sigma_v^2}
\right]^2
\right)
}{
\sqrt{2}\,\sigma_v\sigma_u\sqrt{2\pi}\left(1-\Phi\!\left(-\frac{\mu}{\sigma_u}\right)\right)
}
\\
&\qquad\times
\left[
1-\Phi\!\left(
\sqrt{2}
\left(
\frac{
-\frac{\mu \sigma_v^2-\varepsilon \sigma_u^2}{\sigma_u^2+\sigma_v^2}
}{
\sqrt{\frac{2\sigma_v^2\sigma_u^2}{\sigma_u^2+\sigma_v^2}}
}
\right)
\right)
\right].
\end{align*}
Also,
\[
\int_{0}^{\infty}\exp(-u)\,du=1.
\]
Here \(\Phi(\cdot)\) denotes the cdf of a standard normal random variable.
\end{lemma}

\begin{proof}[Proof of Lemma \ref{LemmaAuxiliar}]
Let \(a_1,a_2,a_3\) be generic nonnegative real numbers, and suppose that \(a_2+a_3\geq 1\). Consider
\begin{align*}
I
&=
\int_{0}^{\infty}
\exp(-a_1 u)
\frac{1}{(\sqrt{2\pi\sigma_v^2})^{a_2}}
\exp\!\left(-\frac{a_2(\varepsilon+u)^2}{2\sigma_v^2}\right)
\frac{1}{(\sqrt{2\pi\sigma_u^2})^{a_3}}
\frac{\exp\!\left(-\frac{a_3(u-\mu)^2}{2\sigma_u^2}\right)}
{\left(1-\Phi\!\left(-\frac{\mu}{\sigma_u}\right)\right)^{a_3}}
\,du.
\end{align*}
This can be rewritten as
\begin{align*}
I
&=
\frac{1}{\sigma_v^{a_2}\sigma_u^{a_3}(\sqrt{2\pi})^{a_2+a_3}\left(1-\Phi\!\left(-\frac{\mu}{\sigma_u}\right)\right)^{a_3}}
\\
&\qquad\times
\int_{0}^{\infty}
\exp\!\left(
-a_1u
-\frac{a_2(\varepsilon+u)^2}{2\sigma_v^2}
-\frac{a_3(u-\mu)^2}{2\sigma_u^2}
\right)\,du.
\end{align*}
Expanding the exponent yields
\begin{align*}
I
&=
\frac{
\exp\!\left(
-\frac{a_2}{2\sigma_v^2}\varepsilon^2
-\frac{a_3}{2\sigma_u^2}\mu^2
\right)
}{
\sigma_v^{a_2}\sigma_u^{a_3}(\sqrt{2\pi})^{a_2+a_3}\left(1-\Phi\!\left(-\frac{\mu}{\sigma_u}\right)\right)^{a_3}
}
\\
&\qquad\times
\int_{0}^{\infty}
\exp\!\left(
-u^2\Big[\frac{\sigma_u^2 a_2+\sigma_v^2 a_3}{2\sigma_v^2\sigma_u^2}\Big]
+
u\Big[\frac{2a_3\mu \sigma_v^2-2a_2\varepsilon\sigma_u^2-2a_1\sigma_v^2\sigma_u^2}{2\sigma_v^2\sigma_u^2}\Big]
\right)\,du.
\end{align*}
Completing the square gives
\begin{align*}
I
&=
\sqrt{\frac{2\sigma_v^2\sigma_u^2}{\sigma_u^2 a_2+\sigma_v^2 a_3}}
\,
\frac{
\exp\!\left(
-\frac{a_2}{2\sigma_v^2}\varepsilon^2
-\frac{a_3}{2\sigma_u^2}\mu^2
+
\frac{\sigma_u^2 a_2+\sigma_v^2 a_3}{2\sigma_v^2\sigma_u^2}
\left[
\frac{a_3\mu \sigma_v^2-a_2\varepsilon \sigma_u^2-a_1\sigma_v^2\sigma_u^2}{\sigma_u^2 a_2+\sigma_v^2 a_3}
\right]^2
\right)
}{
\sqrt{2}\,\sigma_v^{a_2}\sigma_u^{a_3}(\sqrt{2\pi})^{a_2+a_3-1}\left(1-\Phi\!\left(-\frac{\mu}{\sigma_u}\right)\right)^{a_3}
}
\\
&\qquad\times
\left[
1-\Phi\!\left(
\sqrt{2}
\left(
\frac{
-\frac{a_3\mu \sigma_v^2-a_2\varepsilon \sigma_u^2-a_1\sigma_v^2\sigma_u^2}{\sigma_u^2 a_2+\sigma_v^2 a_3}
}{
\sqrt{\frac{2\sigma_v^2\sigma_u^2}{\sigma_u^2 a_2+\sigma_v^2 a_3}}
}
\right)
\right)
\right].
\end{align*}
Applying this expression at the relevant values of \((a_1,a_2,a_3)\) yields the formulas in the statement.
\end{proof}

\subsection{Alternative null models for \(u_i\): Exponential and Half-Normal}\label{AppAltNulls}

This subsection mirrors the revised analysis in Section \ref{sec2:4} by 
specializing the general bounds in Lemma \ref{LemmaIdentifiedset} to two 
alternative null models for the inefficiency term $u_i$ while maintaining 
$v_i \sim N(0, \sigma_v^2)$ and $u_i$ independent of $v_i$. Throughout, 
let $\varepsilon_i = v_i - u_i$, and $\Phi(\cdot)$ denotes the CDF of 
a standard normal random variable.

\subsubsection{Normal--Exponential null: \(u_i\sim Exp(\lambda)\)}

Assume \(u_i\sim Exp(\lambda)\) on \([0,\infty)\) with \(\lambda>0\), and \(v_i\sim N(0,\sigma_v^2)\) independent. Let
\[
\phi_{\sigma_v}(x)=\frac{1}{\sigma_v\sqrt{2\pi}}\exp\!\left(-\frac{x^2}{2\sigma_v^2}\right).
\]
Then the benchmark integral objects are
\begin{align}
\Delta_2(\varepsilon;\sigma_v,\lambda)
&=
\int_{0}^{\infty}\phi_{\sigma_v}(\varepsilon+u)\lambda e^{-\lambda u}\,du
=
\lambda\exp\!\left(\lambda \varepsilon+\frac{1}{2}\lambda^2\sigma_v^2\right)
\Phi\!\left(-\frac{\varepsilon}{\sigma_v}-\lambda\sigma_v\right),
\nonumber\\
\Delta_1(\varepsilon;\sigma_v,\lambda)
&=
\int_{0}^{\infty}e^{-u}\phi_{\sigma_v}(\varepsilon+u)\lambda e^{-\lambda u}\,du
=
\lambda\exp\!\left((\lambda+1)\varepsilon+\frac{1}{2}(\lambda+1)^2\sigma_v^2\right)
\Phi\!\left(-\frac{\varepsilon}{\sigma_v}-(\lambda+1)\sigma_v\right),
\nonumber\\
\Delta_1(\lambda)
&=
\int_{0}^{\infty}e^{-u}\lambda e^{-\lambda u}\,du
=
\frac{\lambda}{\lambda+1}.
\label{DeltasExpClosed}
\end{align}

\begin{lemma}\label{LemmaAuxiliarExp_app}
Under \(v\sim N(0,\sigma_v^2)\), \(u\sim Exp(\lambda)\), and \(v\independent u\), the closed forms in \eqref{DeltasExpClosed} hold.
\end{lemma}

\begin{proof}[Proof of Lemma \ref{LemmaAuxiliarExp_app}]
Start with
\[
\Delta_2(\varepsilon;\sigma_v,\lambda)
=
\int_{0}^{\infty}\phi_{\sigma_v}(\varepsilon+u)\lambda e^{-\lambda u}\,du.
\]
Use the change of variables \(t=\varepsilon+u\), so \(u=t-\varepsilon\) and \(t\in[\varepsilon,\infty)\):
\begin{align*}
\Delta_2(\varepsilon;\sigma_v,\lambda)
&=
\lambda e^{\lambda \varepsilon}\int_{\varepsilon}^{\infty}\phi_{\sigma_v}(t)e^{-\lambda t}\,dt
\\
&=
\lambda e^{\lambda \varepsilon}\int_{\varepsilon}^{\infty}\frac{1}{\sigma_v\sqrt{2\pi}}
\exp\!\left(-\frac{t^2}{2\sigma_v^2}-\lambda t\right)\,dt.
\end{align*}
Complete the square:
\[
-\frac{t^2}{2\sigma_v^2}-\lambda t
=
-\frac{(t+\lambda\sigma_v^2)^2}{2\sigma_v^2}
+\frac{1}{2}\lambda^2\sigma_v^2.
\]
Therefore,
\[
\Delta_2(\varepsilon;\sigma_v,\lambda)
=
\lambda\exp\!\left(\lambda \varepsilon+\frac{1}{2}\lambda^2\sigma_v^2\right)
\int_{\varepsilon}^{\infty}\frac{1}{\sigma_v\sqrt{2\pi}}
\exp\!\left(-\frac{(t+\lambda\sigma_v^2)^2}{2\sigma_v^2}\right)\,dt.
\]
Let \(s=(t+\lambda\sigma_v^2)/\sigma_v\), so \(dt=\sigma_v ds\), and the lower limit is
\[
s_0=\frac{\varepsilon+\lambda\sigma_v^2}{\sigma_v}
=
\frac{\varepsilon}{\sigma_v}+\lambda\sigma_v.
\]
Hence,
\[
\Delta_2(\varepsilon;\sigma_v,\lambda)
=
\lambda\exp\!\left(\lambda \varepsilon+\frac{1}{2}\lambda^2\sigma_v^2\right)
\int_{s_0}^{\infty}\frac{1}{\sqrt{2\pi}}e^{-s^2/2}\,ds
=
\lambda\exp\!\left(\lambda \varepsilon+\frac{1}{2}\lambda^2\sigma_v^2\right)\Phi(-s_0),
\]
which gives the first expression. For \(\Delta_1(\varepsilon;\sigma_v,\lambda)\), note that \(e^{-u}\lambda e^{-\lambda u}=\lambda e^{-(\lambda+1)u}\). Repeating the same derivation with \(\lambda\) replaced by \(\lambda+1\) yields the second expression. Finally,
\[
\Delta_1(\lambda)
=
\int_{0}^{\infty}e^{-u}\lambda e^{-\lambda u}\,du
=
\lambda\int_{0}^{\infty}e^{-(\lambda+1)u}\,du
=
\frac{\lambda}{\lambda+1}.
\]
\end{proof}

Under the revised clipped-envelope formulation, define
\begin{align*}
\underline D_{Exp}(\varepsilon;c,b)
&:=
\int_0^\infty
\big(\phi_{\sigma_v}(\varepsilon+u)-b\big)_+
\big(\lambda e^{-\lambda u}-c\big)_+\,du,
\\
\underline N_{Exp}(\varepsilon;c,b)
&:=
\int_0^\infty
e^{-u}
\big(\phi_{\sigma_v}(\varepsilon+u)-b\big)_+
\big(\lambda e^{-\lambda u}-c\big)_+\,du.
\end{align*}
The corresponding upper-envelope objects are
\begin{align*}
\overline D_{Exp}(\varepsilon;c,b)
&=
\Delta_2(\varepsilon;\sigma_v,\lambda)+c\,\Delta_2(\varepsilon;\sigma_v)+b,
\\
\overline N_{Exp}(\varepsilon;c,b)
&=
\Delta_1(\varepsilon;\sigma_v,\lambda)+c\,\Delta_1(\varepsilon;\sigma_v)+b\,\Delta_1(\lambda)+bc.
\end{align*}
Therefore,
\[
\frac{\underline N_{Exp}(\varepsilon;c,b)}{\overline D_{Exp}(\varepsilon;c,b)}
\leq
E[\exp(-u_i)\mid \varepsilon_i=\varepsilon]
\leq
\frac{\overline N_{Exp}(\varepsilon;c,b)}{\underline D_{Exp}(\varepsilon;c,b)},
\]
provided \(\underline D_{Exp}(\varepsilon;c,b)>0\).

For inference, let \(S_{\kappa}\) denote the smooth clipping operator defined in Appendix \ref{AppSmoothClip}, and define
\begin{align}
\underline D_{Exp,\kappa}(\varepsilon;c,b)
&:=
\int_0^\infty
S_{\kappa}\!\big(\phi_{\sigma_v}(\varepsilon+u);b\big)
S_{\kappa}\!\big(\lambda e^{-\lambda u};c\big)\,du,
\nonumber\\
\underline N_{Exp,\kappa}(\varepsilon;c,b)
&:=
\int_0^\infty
e^{-u}
S_{\kappa}\!\big(\phi_{\sigma_v}(\varepsilon+u);b\big)
S_{\kappa}\!\big(\lambda e^{-\lambda u};c\big)\,du.
\label{EqExpSmoothObjects}
\end{align}
Then, for fixed \((c,\tau_0)\), the smoothed breakdown frontier \(b_{\kappa}^{Exp}(c,\tau_0)\) is defined implicitly by
\begin{equation}\label{EqExpFrontierImplicit}
G_{\kappa}^{Exp}(b;c,\tau_0,\vartheta)
:=
\underline N_{Exp,\kappa}(\varepsilon;c,b)
-
\tau_0\,\overline D_{Exp}(\varepsilon;c,b)
=0.
\end{equation}

\subsubsection{Normal--Half-Normal null: \(u_i\sim HN(\sigma_u^2)\)}

Assume \(u_i\sim HN(\sigma_u^2)\) on \([0,\infty)\), with density
\[
f_u(u)=\sqrt{\frac{2}{\pi}}\frac{1}{\sigma_u}\exp\!\left(-\frac{u^2}{2\sigma_u^2}\right),
\]
and \(v_i\sim N(0,\sigma_v^2)\) independent. Define
\[
\sigma^2=\sigma_v^2+\sigma_u^2,
\qquad
\kappa_0=\frac{\sigma_u}{\sigma_v\sigma}.
\]
Then the benchmark integral objects are
\begin{align}
\Delta_2(\varepsilon;\sigma_v,\sigma_u)
&=
\int_{0}^{\infty}
\phi_{\sigma_v}(\varepsilon+u)\sqrt{\frac{2}{\pi}}\frac{1}{\sigma_u}
\exp\!\left(-\frac{u^2}{2\sigma_u^2}\right)\,du
=
2\frac{1}{\sqrt{2\pi\sigma^2}}
\exp\!\left(-\frac{\varepsilon^2}{2\sigma^2}\right)\Phi(-\kappa_0 \varepsilon),
\nonumber\\
\Delta_1(\varepsilon;\sigma_v,\sigma_u)
&=
\int_{0}^{\infty}
e^{-u}\phi_{\sigma_v}(\varepsilon+u)\sqrt{\frac{2}{\pi}}\frac{1}{\sigma_u}
\exp\!\left(-\frac{u^2}{2\sigma_u^2}\right)\,du
\nonumber\\
&=
2\frac{1}{\sqrt{2\pi\sigma^2}}
\exp\!\left(-\frac{\varepsilon^2}{2\sigma^2}\right)
\exp\!\left(\frac{\sigma_u^2}{\sigma^2}\varepsilon+\frac{\sigma_u^2\sigma_v^2}{2\sigma^2}\right)
\Phi\!\left(-\kappa_0 \varepsilon-\frac{\sigma_u\sigma_v}{\sigma}\right),
\nonumber\\
\Delta_1(\sigma_u)
&=
\int_{0}^{\infty}
e^{-u}\sqrt{\frac{2}{\pi}}\frac{1}{\sigma_u}
\exp\!\left(-\frac{u^2}{2\sigma_u^2}\right)\,du
=
2\exp\!\left(\frac{1}{2}\sigma_u^2\right)\Phi(-\sigma_u).
\label{DeltasHNClosed}
\end{align}

\begin{lemma}\label{LemmaAuxiliarHalfNormal_app}
Under \(v\sim N(0,\sigma_v^2)\), \(u\sim HN(\sigma_u^2)\), and \(v\independent u\), the closed forms in \eqref{DeltasHNClosed} hold.
\end{lemma}

\begin{proof}[Proof of Lemma \ref{LemmaAuxiliarHalfNormal_app}]
Write
\[
\Delta_2(\varepsilon;\sigma_v,\sigma_u)=\int_{0}^{\infty}\phi_{\sigma_v}(\varepsilon+u)f_u(u)\,du,
\]
with \(f_u\) half-normal. Using \(\phi_{\sigma_v}(\varepsilon+u)\propto\exp(-(\varepsilon+u)^2/(2\sigma_v^2))\) and \(f_u(u)\propto\exp(-u^2/(2\sigma_u^2))\), the exponent in \(u\) is
\[
-\frac{(\varepsilon+u)^2}{2\sigma_v^2}-\frac{u^2}{2\sigma_u^2}
=
-\frac{\varepsilon^2}{2\sigma_v^2}
-\frac{\sigma^2}{2\sigma_v^2\sigma_u^2}u^2
-\frac{\varepsilon}{\sigma_v^2}u.
\]
Let
\[
A=\frac{\sigma^2}{2\sigma_v^2\sigma_u^2},
\qquad
B=\frac{\varepsilon}{\sigma_v^2}.
\]
Completing the square gives
\[
-Au^2-Bu
=
-A\left(u+\frac{B}{2A}\right)^2+\frac{B^2}{4A},
\qquad
\frac{B}{2A}=\frac{\varepsilon\sigma_u^2}{\sigma^2}.
\]
Moreover,
\[
-\frac{\varepsilon^2}{2\sigma_v^2}+\frac{B^2}{4A}
=
-\frac{\varepsilon^2}{2\sigma^2}.
\]
Hence
\[
\Delta_2(\varepsilon;\sigma_v,\sigma_u)
=
K\exp\!\left(-\frac{\varepsilon^2}{2\sigma^2}\right)
\int_{0}^{\infty}\exp\!\left(-A\left(u+\frac{\varepsilon\sigma_u^2}{\sigma^2}\right)^2\right)\,du,
\]
where \(K=\sqrt{2/\pi}\,(\sigma_v\sigma_u\sqrt{2\pi})^{-1}\). Set
\[
w=\sqrt{2A}\left(u+\frac{\varepsilon\sigma_u^2}{\sigma^2}\right),
\]
so \(du=(1/\sqrt{2A})\,dw\). At \(u=0\),
\[
w_0=\sqrt{2A}\frac{\varepsilon\sigma_u^2}{\sigma^2}
=
\kappa_0 \varepsilon.
\]
Therefore,
\[
\int_{0}^{\infty}\exp\!\left(-A\left(u+\frac{\varepsilon\sigma_u^2}{\sigma^2}\right)^2\right)\,du
=
\frac{\sqrt{2\pi}}{\sqrt{2A}}\Phi(-w_0)
=
\frac{\sqrt{2\pi}\,\sigma_v\sigma_u}{\sigma}\Phi(-\kappa_0 \varepsilon).
\]
Multiplying constants yields the first expression. The remaining two expressions follow by the same square-completion argument, with the linear term shifted by \(e^{-u}\) when computing \(\Delta_1(\varepsilon;\sigma_v,\sigma_u)\), and directly in the one-dimensional half-normal integral for \(\Delta_1(\sigma_u)\).
\end{proof}

Under the revised clipped-envelope formulation, define
\begin{align*}
\underline D_{HN}(\varepsilon;c,b)
&:=
\int_0^\infty
\big(\phi_{\sigma_v}(\varepsilon+u)-b\big)_+
\left(
\sqrt{\frac{2}{\pi}}\frac{1}{\sigma_u}\exp\!\left(-\frac{u^2}{2\sigma_u^2}\right)-c
\right)_+du,
\\
\underline N_{HN}(\varepsilon;c,b)
&:=
\int_0^\infty
e^{-u}
\big(\phi_{\sigma_v}(\varepsilon+u)-b\big)_+
\left(
\sqrt{\frac{2}{\pi}}\frac{1}{\sigma_u}\exp\!\left(-\frac{u^2}{2\sigma_u^2}\right)-c
\right)_+du.
\end{align*}
The corresponding upper-envelope objects are
\begin{align*}
\overline D_{HN}(\varepsilon;c,b)
&=
\Delta_2(\varepsilon;\sigma_v,\sigma_u)+c\,\Delta_2(\varepsilon;\sigma_v)+b,
\\
\overline N_{HN}(\varepsilon;c,b)
&=
\Delta_1(\varepsilon;\sigma_v,\sigma_u)+c\,\Delta_1(\varepsilon;\sigma_v)+b\,\Delta_1(\sigma_u)+bc.
\end{align*}
Therefore,
\[
\frac{\underline N_{HN}(\varepsilon;c,b)}{\overline D_{HN}(\varepsilon;c,b)}
\leq
E[\exp(-u_i)\mid \varepsilon_i=\varepsilon]
\leq
\frac{\overline N_{HN}(\varepsilon;c,b)}{\underline D_{HN}(\varepsilon;c,b)},
\]
provided \(\underline D_{HN}(\varepsilon;c,b)>0\).

For inference, define the smoothed lower-envelope objects
\begin{align}
\underline D_{HN,\kappa}(\varepsilon;c,b)
&:=
\int_0^\infty
S_{\kappa}\!\big(\phi_{\sigma_v}(\varepsilon+u);b\big)
S_{\kappa}\!\left(
\sqrt{\frac{2}{\pi}}\frac{1}{\sigma_u}\exp\!\left(-\frac{u^2}{2\sigma_u^2}\right);c
\right)\,du,
\nonumber\\
\underline N_{HN,\kappa}(\varepsilon;c,b)
&:=
\int_0^\infty
e^{-u}
S_{\kappa}\!\big(\phi_{\sigma_v}(\varepsilon+u);b\big)
S_{\kappa}\!\left(
\sqrt{\frac{2}{\pi}}\frac{1}{\sigma_u}\exp\!\left(-\frac{u^2}{2\sigma_u^2}\right);c
\right)\,du.
\label{EqHNSmoothObjects}
\end{align}
Then, for fixed \((c,\tau_0)\), the smoothed breakdown frontier \(b_{\kappa}^{HN}(c,\tau_0)\) is defined implicitly by
\begin{equation}\label{EqHNFrontierImplicit}
G_{\kappa}^{HN}(b;c,\tau_0,\vartheta)
:=
\underline N_{HN,\kappa}(\varepsilon;c,b)
-
\tau_0\,\overline D_{HN}(\varepsilon;c,b)
=0.
\end{equation}

\subsection{The Normal--Exponential and Normal--Half-Normal models continued: estimation and inference procedures for the breakdown frontier}\label{AppAltNulls_25_detail}

This subsection replicates the revised inference strategy from Section 2.5 for the two alternative null models for \(u_i\), while maintaining \(v_i\sim N(0,\sigma_v^2)\) and \(v_i\independent u_i\). As in the main text, let
\[
\varepsilon_i \equiv y_i-f(\theta_x,x_i),
\]
and let \(\vartheta\) collect all parameters, including \(\theta_x\) and the nuisance distributional parameters. Throughout, \(\Phi(\cdot)\) and \(\phi(\cdot)\) denote the standard normal cdf and pdf, and we use the inverse Mills ratio
\[
\Lambda(z)\equiv \frac{\phi(z)}{\Phi(z)},
\qquad
\Lambda'(z)=-\Lambda(z)\big(z+\Lambda(z)\big).
\]
We also use the smoothing operators \(S_\kappa\) and \(B_\eta\) defined in Appendix \ref{AppSmoothClip}.

\subsubsection{Case A: \(v_i\sim N(0,\sigma_v^2)\) and \(u_i\sim Exp(\lambda)\)}\label{AppAltNulls_25_exp}

Assume \(u_i\sim Exp(\lambda)\) on \([0,\infty)\) with \(\lambda>0\), and \(v_i\sim N(0,\sigma_v^2)\) independent. Then the composite-error density is
\begin{equation}\label{fu_exp_app}
f_{\varepsilon}(\varepsilon;\sigma_v,\lambda)
=
\lambda \exp\!\left(\lambda \varepsilon+\frac{1}{2}\lambda^2\sigma_v^2\right)
\Phi\!\left(z(\varepsilon;\sigma_v,\lambda)\right),
\qquad
z(\varepsilon;\sigma_v,\lambda)\equiv -\frac{\varepsilon}{\sigma_v}-\lambda\sigma_v.
\end{equation}

\paragraph{Log-likelihood.}
For an i.i.d.\ sample, the log-likelihood is
\begin{equation}\label{LogLikelihood_exp_app}
LnL(\vartheta)=\sum_{i=1}^n \ell_i(\vartheta),
\qquad
\ell_i(\vartheta)=\ln\lambda+\lambda \varepsilon_i+\frac{1}{2}\lambda^2\sigma_v^2+\ln\Phi(z_i),
\quad
z_i\equiv z(\varepsilon_i;\sigma_v,\lambda).
\end{equation}

\paragraph{Score and Hessian.}
Define
\[
g_i\equiv \frac{\partial \varepsilon_i}{\partial \theta_x}
=
-\frac{\partial f(\theta_x,x_i)}{\partial \theta_x},
\qquad
G_i\equiv \frac{\partial^2 \varepsilon_i}{\partial \theta_x\partial \theta_x'}.
\]
Using
\[
\frac{\partial z_i}{\partial \varepsilon_i}=-\frac{1}{\sigma_v},
\qquad
\frac{\partial z_i}{\partial \sigma_v}=\frac{\varepsilon_i}{\sigma_v^2}-\lambda,
\qquad
\frac{\partial z_i}{\partial \lambda}=-\sigma_v,
\]
we obtain
\begin{align}
\frac{\partial \ell_i}{\partial \varepsilon_i}
&=
\lambda-\frac{1}{\sigma_v}\Lambda(z_i),
\nonumber\\
\frac{\partial \ell_i}{\partial \theta_x}
&=
-\left(\lambda-\frac{1}{\sigma_v}\Lambda(z_i)\right)g_i,
\nonumber\\
\frac{\partial \ell_i}{\partial \sigma_v}
&=
\lambda^2\sigma_v+\Lambda(z_i)\left(\frac{\varepsilon_i}{\sigma_v^2}-\lambda\right),
\nonumber\\
\frac{\partial \ell_i}{\partial \lambda}
&=
\frac{1}{\lambda}+\varepsilon_i+\lambda\sigma_v^2-\sigma_v\Lambda(z_i),
\label{score_exp_app}
\end{align}
and the scalar second derivatives
\begin{align}
\frac{\partial^2 \ell_i}{\partial \varepsilon_i^2}
&=
\frac{\Lambda'(z_i)}{\sigma_v^2},
\nonumber\\
\frac{\partial^2 \ell_i}{\partial \varepsilon_i\,\partial \sigma_v}
&=
-\frac{1}{\sigma_v}\Lambda'(z_i)\left(\frac{\varepsilon_i}{\sigma_v^2}-\lambda\right)
+\frac{1}{\sigma_v^2}\Lambda(z_i),
\nonumber\\
\frac{\partial^2 \ell_i}{\partial \varepsilon_i\,\partial \lambda}
&=
1+\Lambda'(z_i),
\nonumber\\
\frac{\partial^2 \ell_i}{\partial \sigma_v^2}
&=
\lambda^2+\Lambda'(z_i)\left(\frac{\varepsilon_i}{\sigma_v^2}-\lambda\right)^2
+\Lambda(z_i)\left(-\frac{2\varepsilon_i}{\sigma_v^3}\right),
\nonumber\\
\frac{\partial^2 \ell_i}{\partial \lambda^2}
&=
-\frac{1}{\lambda^2}+\sigma_v^2+\sigma_v^2\Lambda'(z_i),
\nonumber\\
\frac{\partial^2 \ell_i}{\partial \sigma_v\,\partial \lambda}
&=
2\lambda\sigma_v+\Lambda'(z_i)\left(\frac{\varepsilon_i}{\sigma_v^2}-\lambda\right)(-\sigma_v)-\Lambda(z_i).
\label{second_exp_scalar}
\end{align}

Under standard regularity conditions,
\[
\sqrt{n}(\widehat{\vartheta}_n-\vartheta)\xrightarrow{d}N(0,\mathcal{V}_{\vartheta}),
\]
with \(\mathcal{V}_{\vartheta}\) estimated by the sandwich form described in the main text.

For each fixed \((\kappa,\eta)\), define the smoothed frontier equation
\[
\widetilde G_{\kappa,\eta}^{Exp}(\beta;c,\tau_0,\vartheta)
:=
G_{\kappa}^{Exp}(B_\eta(\beta);c,\tau_0,\vartheta),
\]
where \(G_{\kappa}^{Exp}\) is defined in \eqref{EqExpFrontierImplicit}. Let \(\beta_{\kappa,\eta}^{Exp}(c,\tau_0)\) denote the unique solution in \(\mathbb R\) to
\[
\widetilde G_{\kappa,\eta}^{Exp}(\beta;c,\tau_0,\vartheta)=0,
\]
and define
\[
b_{\kappa,\eta}^{Exp}(c,\tau_0)
=
B_\eta\big(\beta_{\kappa,\eta}^{Exp}(c,\tau_0)\big).
\]
Its plug-in estimator is obtained by solving
\[
\widetilde G_{\kappa_n,\eta_n}^{Exp}(\widehat\beta_{n,\kappa_n,\eta_n}^{Exp}(c,\tau_0);c,\tau_0,\widehat{\vartheta}_n)=0,
\]
and setting
\[
\widehat b_{n,\kappa_n,\eta_n}^{Exp}(c,\tau_0)
=
B_{\eta_n}\big(\widehat\beta_{n,\kappa_n,\eta_n}^{Exp}(c,\tau_0)\big).
\]

If
\[
\frac{\partial}{\partial \beta}
\widetilde G_{\kappa,\eta}^{Exp}(\beta_{\kappa,\eta}^{Exp}(c,\tau_0);c,\tau_0,\vartheta)\neq 0,
\]
then
\begin{equation}\label{Asy_beta_exp_detail_new}
\sqrt{n}\big(\widehat\beta_{n,\kappa,\eta}^{Exp}(c,\tau_0)-\beta_{\kappa,\eta}^{Exp}(c,\tau_0)\big)
\xrightarrow{d}
N\big(0,\sigma_{\beta_{\kappa,\eta}^{Exp}(c,\tau_0)}^2\big),
\end{equation}
where
\begin{equation}\label{Var_beta_exp_app_new}
\sigma_{\beta_{\kappa,\eta}^{Exp}(c,\tau_0)}^2
=
D_{\beta_{\kappa,\eta}^{Exp}}\,\mathcal{V}_{\vartheta}\,D_{\beta_{\kappa,\eta}^{Exp}}',
\qquad
D_{\beta_{\kappa,\eta}^{Exp}}
=
-
\left(
\frac{\partial}{\partial \beta}
\widetilde G_{\kappa,\eta}^{Exp}(\beta_{\kappa,\eta}^{Exp};c,\tau_0,\vartheta)
\right)^{-1}
\frac{\partial}{\partial \vartheta'}
\widetilde G_{\kappa,\eta}^{Exp}(\beta_{\kappa,\eta}^{Exp};c,\tau_0,\vartheta).
\end{equation}

Applying the delta method through \(b_{\kappa,\eta}^{Exp}=B_\eta(\beta_{\kappa,\eta}^{Exp})\), we obtain
\begin{equation}\label{Asy_b_exp_detail_new}
\sqrt{n}\big(\widehat b_{n,\kappa,\eta}^{Exp}(c,\tau_0)-b_{\kappa,\eta}^{Exp}(c,\tau_0)\big)
\xrightarrow{d}
N\big(0,\sigma_{b_{\kappa,\eta}^{Exp}(c,\tau_0)}^2\big),
\end{equation}
where
\begin{equation}\label{Var_b_exp_app_new}
\sigma_{b_{\kappa,\eta}^{Exp}(c,\tau_0)}^2
=
D_{b_{\kappa,\eta}^{Exp}}\,\mathcal{V}_{\vartheta}\,D_{b_{\kappa,\eta}^{Exp}}',
\qquad
D_{b_{\kappa,\eta}^{Exp}}
=
B_\eta'\big(\beta_{\kappa,\eta}^{Exp}(c,\tau_0)\big)\,
D_{\beta_{\kappa,\eta}^{Exp}}.
\end{equation}
A feasible estimator is obtained by plugging in \(\widehat{\vartheta}_n\), \(\widehat{\mathcal{V}}_{\vartheta,n}\), and \(\widehat\beta_{n,\kappa_n,\eta_n}^{Exp}(c,\tau_0)\).

\subsubsection{Case B: \(v_i\sim N(0,\sigma_v^2)\) and \(u_i\sim HN(\sigma_u^2)\)}\label{AppAltNulls_25_hn}

Assume \(u_i\sim HN(\sigma_u^2)\) on \([0,\infty)\), with density
\[
f_u(u)=\sqrt{\frac{2}{\pi}}\frac{1}{\sigma_u}\exp\!\left(-\frac{u^2}{2\sigma_u^2}\right),
\]
and \(v_i\sim N(0,\sigma_v^2)\) independent. Define
\[
\sigma\equiv\sqrt{\sigma_v^2+\sigma_u^2},
\qquad
\kappa_0\equiv \frac{\sigma_u}{\sigma_v\sigma}.
\]
Then the composite-error density is
\begin{equation}\label{fu_hn_app}
f_{\varepsilon}(\varepsilon;\sigma_v,\sigma_u)
=
2\frac{1}{\sqrt{2\pi\sigma^2}}
\exp\!\left(-\frac{\varepsilon^2}{2\sigma^2}\right)
\Phi\!\left(w(\varepsilon;\sigma_v,\sigma_u)\right),
\qquad
w(\varepsilon;\sigma_v,\sigma_u)\equiv -\kappa_0 \varepsilon.
\end{equation}

\paragraph{Log-likelihood.}
\begin{equation}\label{LogLikelihood_hn_app}
LnL(\vartheta)=\sum_{i=1}^n \ell_i(\vartheta),
\qquad
\ell_i(\vartheta)
=
\ln 2-\frac{1}{2}\ln(2\pi)-\ln \sigma-\frac{\varepsilon_i^2}{2\sigma^2}+\ln\Phi(w_i),
\quad
w_i\equiv -\kappa_0 \varepsilon_i.
\end{equation}

\paragraph{Score and Hessian.}
Let \(g_i=\partial \varepsilon_i/\partial \theta_x\) and \(G_i=\partial^2 \varepsilon_i/\partial \theta_x\partial \theta_x'\). Let \(\Lambda(w)=\phi(w)/\Phi(w)\). Then
\begin{align}
\frac{\partial \ell_i}{\partial \varepsilon_i}
&=
-\frac{\varepsilon_i}{\sigma^2}-\kappa_0 \Lambda(w_i),
\nonumber\\
\frac{\partial \ell_i}{\partial \theta_x}
&=
-\left(-\frac{\varepsilon_i}{\sigma^2}-\kappa_0 \Lambda(w_i)\right)g_i,
\nonumber\\
\frac{\partial \ell_i}{\partial \sigma_v}
&=
\left(-\frac{1}{\sigma}+\frac{\varepsilon_i^2}{\sigma^3}\right)\frac{\partial\sigma}{\partial\sigma_v}
+\Lambda(w_i)\frac{\partial w_i}{\partial\sigma_v},
\nonumber\\
\frac{\partial \ell_i}{\partial \sigma_u}
&=
\left(-\frac{1}{\sigma}+\frac{\varepsilon_i^2}{\sigma^3}\right)\frac{\partial\sigma}{\partial\sigma_u}
+\Lambda(w_i)\frac{\partial w_i}{\partial\sigma_u},
\label{score_hn_app}
\end{align}
with
\begin{align}
\frac{\partial \sigma}{\partial \sigma_v}
&=
\frac{\sigma_v}{\sigma},
\qquad
\frac{\partial \sigma}{\partial \sigma_u}
=
\frac{\sigma_u}{\sigma},
\nonumber\\
\frac{\partial \kappa_0}{\partial \sigma_v}
&=
\kappa_0\left(-\frac{1}{\sigma_v}-\frac{\sigma_v}{\sigma^2}\right),
\qquad
\frac{\partial \kappa_0}{\partial \sigma_u}
=
\kappa_0\left(\frac{1}{\sigma_u}-\frac{\sigma_u}{\sigma^2}\right),
\nonumber\\
\frac{\partial w_i}{\partial \sigma_v}
&=
-\varepsilon_i\frac{\partial \kappa_0}{\partial \sigma_v},
\qquad
\frac{\partial w_i}{\partial \sigma_u}
=
-\varepsilon_i\frac{\partial \kappa_0}{\partial \sigma_u},
\qquad
\frac{\partial w_i}{\partial \varepsilon_i}
=
-\kappa_0.
\label{hn_primitives}
\end{align}
Under standard regularity conditions,
\[
\sqrt{n}(\widehat{\vartheta}_n-\vartheta)\xrightarrow{d}N(0,\mathcal{V}_{\vartheta}),
\]
again with sandwich variance estimator.

For each fixed \((\kappa,\eta)\), define the smoothed frontier equation
\[
\widetilde G_{\kappa,\eta}^{HN}(\beta;c,\tau_0,\vartheta)
:=
G_{\kappa}^{HN}(B_\eta(\beta);c,\tau_0,\vartheta),
\]
where \(G_{\kappa}^{HN}\) is defined in \eqref{EqHNFrontierImplicit}. Let \(\beta_{\kappa,\eta}^{HN}(c,\tau_0)\) denote the unique solution in \(\mathbb R\) to
\[
\widetilde G_{\kappa,\eta}^{HN}(\beta;c,\tau_0,\vartheta)=0,
\]
and define
\[
b_{\kappa,\eta}^{HN}(c,\tau_0)
=
B_\eta\big(\beta_{\kappa,\eta}^{HN}(c,\tau_0)\big).
\]
Its plug-in estimator is obtained by solving
\[
\widetilde G_{\kappa_n,\eta_n}^{HN}(\widehat\beta_{n,\kappa_n,\eta_n}^{HN}(c,\tau_0);c,\tau_0,\widehat{\vartheta}_n)=0,
\]
and setting
\[
\widehat b_{n,\kappa_n,\eta_n}^{HN}(c,\tau_0)
=
B_{\eta_n}\big(\widehat\beta_{n,\kappa_n,\eta_n}^{HN}(c,\tau_0)\big).
\]

If
\[
\frac{\partial}{\partial \beta}
\widetilde G_{\kappa,\eta}^{HN}(\beta_{\kappa,\eta}^{HN}(c,\tau_0);c,\tau_0,\vartheta)\neq 0,
\]
then
\begin{equation}\label{Asy_beta_hn_detail_new}
\sqrt{n}\big(\widehat\beta_{n,\kappa,\eta}^{HN}(c,\tau_0)-\beta_{\kappa,\eta}^{HN}(c,\tau_0)\big)
\xrightarrow{d}
N\big(0,\sigma_{\beta_{\kappa,\eta}^{HN}(c,\tau_0)}^2\big),
\end{equation}
where
\begin{equation}\label{Var_beta_hn_app_new}
\sigma_{\beta_{\kappa,\eta}^{HN}(c,\tau_0)}^2
=
D_{\beta_{\kappa,\eta}^{HN}}\,\mathcal{V}_{\vartheta}\,D_{\beta_{\kappa,\eta}^{HN}}',
\qquad
D_{\beta_{\kappa,\eta}^{HN}}
=
-
\left(
\frac{\partial}{\partial \beta}
\widetilde G_{\kappa,\eta}^{HN}(\beta_{\kappa,\eta}^{HN};c,\tau_0,\vartheta)
\right)^{-1}
\frac{\partial}{\partial \vartheta'}
\widetilde G_{\kappa,\eta}^{HN}(\beta_{\kappa,\eta}^{HN};c,\tau_0,\vartheta).
\end{equation}

Applying the delta method through \(b_{\kappa,\eta}^{HN}=B_\eta(\beta_{\kappa,\eta}^{HN})\), we obtain
\begin{equation}\label{Asy_b_hn_detail_new}
\sqrt{n}\big(\widehat b_{n,\kappa,\eta}^{HN}(c,\tau_0)-b_{\kappa,\eta}^{HN}(c,\tau_0)\big)
\xrightarrow{d}
N\big(0,\sigma_{b_{\kappa,\eta}^{HN}(c,\tau_0)}^2\big),
\end{equation}
where
\begin{equation}\label{Var_b_hn_app_new}
\sigma_{b_{\kappa,\eta}^{HN}(c,\tau_0)}^2
=
D_{b_{\kappa,\eta}^{HN}}\,\mathcal{V}_{\vartheta}\,D_{b_{\kappa,\eta}^{HN}}',
\qquad
D_{b_{\kappa,\eta}^{HN}}
=
B_\eta'\big(\beta_{\kappa,\eta}^{HN}(c,\tau_0)\big)\,
D_{\beta_{\kappa,\eta}^{HN}}.
\end{equation}
A feasible estimator is obtained by plugging in \(\widehat{\vartheta}_n\), \(\widehat{\mathcal{V}}_{\vartheta,n}\), and \(\widehat\beta_{n,\kappa_n,\eta_n}^{HN}(c,\tau_0)\).

\subsection{Details of the smoothing operators}\label{AppSmoothClip}

Recall first the clipping-smoothing function
\[
S_{\kappa}(x;a)
=
\kappa\log\!\left(1+\exp\!\left(\frac{x-a}{\kappa}\right)\right)
-
\kappa\log\!\left(1+\exp\!\left(\frac{-a}{\kappa}\right)\right),
\qquad \kappa>0,
\]
for \(x\geq 0\) and \(a\geq 0\). This function satisfies \(S_{\kappa}(0;a)=0\), \(S_{\kappa}(x;a)\geq 0\), and
\[
S_{\kappa}(x;a)\to (x-a)_+
\qquad\text{as}\qquad \kappa\downarrow 0.
\]

Its first derivatives are
\begin{align*}
\frac{\partial}{\partial x}S_{\kappa}(x;a)
&=
\frac{1}{1+\exp\!\left(-\frac{x-a}{\kappa}\right)},
\\
\frac{\partial}{\partial a}S_{\kappa}(x;a)
&=
-
\frac{1}{1+\exp\!\left(-\frac{x-a}{\kappa}\right)}
+
\frac{1}{1+\exp\!\left(\frac{a}{\kappa}\right)}.
\end{align*}
Its second derivative with respect to \(x\) is
\[
\frac{\partial^2}{\partial x^2}S_{\kappa}(x;a)
=
\frac{1}{\kappa}
\frac{\exp\!\left(-\frac{x-a}{\kappa}\right)}
{\left(1+\exp\!\left(-\frac{x-a}{\kappa}\right)\right)^2}.
\]
Moreover,
\[
\lim_{\kappa\downarrow 0}\frac{\partial}{\partial x}S_{\kappa}(x;a)
=
\mathbf{1}\{x>a\}
\qquad\text{for every }x\neq a.
\]

Next, recall the smooth positivity map
\[
B_{\eta}(\beta)
=
\eta\log\!\left(1+\exp\!\left(\frac{\beta}{\eta}\right)\right),
\qquad \eta>0,
\]
defined for \(\beta\in\mathbb R\). This function satisfies \(B_{\eta}(\beta)>0\) for every \(\beta\in\mathbb R\), and
\[
B_{\eta}(\beta)\to \beta_+
\qquad\text{as}\qquad \eta\downarrow 0.
\]

Its first derivative is
\[
\frac{\partial}{\partial \beta}B_{\eta}(\beta)
=
\frac{1}{1+\exp(-\beta/\eta)},
\]
and its second derivative is
\[
\frac{\partial^2}{\partial \beta^2}B_{\eta}(\beta)
=
\frac{1}{\eta}
\frac{\exp(-\beta/\eta)}
{\left(1+\exp(-\beta/\eta)\right)^2}.
\]
Moreover,
\[
\lim_{\eta\downarrow 0}\frac{\partial}{\partial \beta}B_{\eta}(\beta)
=
\mathbf{1}\{\beta>0\}
\qquad\text{for every }\beta\neq 0.
\]

\end{document}